\PassOptionsToPackage{prologue, table}{xcolor}
\documentclass[sigconf]{acmart}

\usepackage{graphicx}
\usepackage{amsmath}
\usepackage[ruled,vlined,linesnumbered]{algorithm2e}
\usepackage{subfigure}
\usepackage{multirow}
\usepackage{url}
\usepackage{mathrsfs}
\usepackage{hyperref}
\usepackage{amsthm}

\usepackage{amssymb}
\usepackage{color}
\usepackage[table]{xcolor}
\usepackage{enumitem}
\usepackage{calc}
\usepackage[columnwise,displaymath,mathlines]{lineno}
\newtheorem{theorem}{\bf Theorem}

\newtheorem{definition}{\bf Definition}
\newtheorem{prop}{\bf Proposition}

\usepackage{array}

\definecolor{myGreen}{RGB}{13, 134, 13} 
\begin{document}
	\title{An Efficient and Exact Algorithm for Locally $h$-Clique Densest Subgraph Discovery}
	
	\author{Xiaojia Xu}
	\affiliation{%
		\institution{Renmin University of China}
		\city{Beijing}
		\country{China}}
	\email{xuxiaojia@ruc.edu.cn}
	
	\author{Haoyu Liu}
	\affiliation{%
		\institution{Renmin University of China}
		\city{Beijing}
		\country{China}}
	\email{hyliu2187@ruc.edu.cn}
	
	\author{Xiaowei Lv}
	\affiliation{%
		\institution{Renmin University of China}
		\city{Beijing}
		\country{China}}
	\email{lvxiaowei@ruc.edu.cn}
	
	\author{Yongcai Wang$^{\ast}$}
	\affiliation{%
		\institution{Renmin University of China}
		\city{Beijing}
		\country{China}}
	\email{ycw@ruc.edu.cn}
	
	\author{Deying Li}
	\affiliation{%
		\institution{Renmin University of China}
		\city{Beijing}
		\country{China}}
	\email{deyingli@ruc.edu.cn}

\begin{abstract}
	Detecting locally, non-overlapping, near-clique densest subgraphs is a crucial problem for community search in social networks. As a vertex may be involved in multiple overlapped local cliques, detecting locally densest sub-structures considering $h$-clique density, i.e., \emph{locally $h$-clique densest subgraph (L$h$CDS)}  attracts great interests. This paper investigates the L$h$CDS detection problem and proposes an efficient and exact algorithm to list the top-$k$ non-overlapping, locally $h$-clique dense, and compact subgraphs. 
	We  in particular jointly consider $h$-clique compact number and L$h$CDS  and design a new ``Iterative Propose-Prune-and-Verify'' pipeline (\texttt{IPPV}) for top-$k$ L$h$CDS detection. 
	(1) In the proposal part, we derive initial bounds for $h$-clique compact numbers; prove the validity, and extend a convex programming method to tighten the bounds for proposing L$h$CDS candidates without missing any. 
	(2) Then a tentative graph decomposition method is proposed to solve the challenging case where a clique spans multiple subgraphs in graph decomposition. 
	(3) To deal with the verification difficulty, both a basic and a fast verification method are proposed, where the fast method constructs a smaller-scale flow network to improve efficiency while preserving the verification correctness. The verified L$h$CDSes are returned, while the candidates that remained unsure reenter the \texttt{IPPV} pipeline. 
	(4) We further extend the proposed methods to locally more general pattern densest subgraph detection problems.
	We prove the exactness and low complexity of the proposed algorithm. Extensive experiments on real datasets show the effectiveness and high efficiency of \texttt{IPPV}. 
\end{abstract}

\maketitle

\section{Introduction}
Finding dense subgraphs can uncover highly connected and cohesive structures in graphs, making it an effective tool for understanding complex systems. The discovery of dense subgraphs and communities has numerous applications in diverse fields including social networks \cite{Chen2012DenseSE,Tsourakakis2013DenserTT,Gibson2005DiscoveringLD}, web analysis \cite{Gionis2013PiggybackingOS,Angel2012DenseSM}, graph databases \cite{Jin20093HOPAH,Zhao2012LargeSC}, and biology \cite{Saha2010DenseSW,Li2022DensestSM}.  In these applications, the identification of near-clique subgraphs holds significant importance, as it relaxes the requirement of complete connectivity within cliques and allows for a certain degree of sparsity or missing connections while still maintaining a high level of connectivity. 

Given the importance of detecting large near-clique subgraphs \cite{tsourakakis2015k}, 
the $h$-clique densest subgraph (CDS) problem that finds near-cliques formed by overlapping cliques has attracted great research attention \cite{tsourakakis2015k,Mitzenmacher2015ScalableLN,Fang2019EfficientAF,sun2020kclist++}. 
This is due to the fact that a vertex is generally involved in multiple overlapping cliques, such as a person may be involved in cliques as family members, office mates, etc \cite{Palla2005UncoveringTO}. By finding the subgraph with the highest density of $h$-cliques, CDS uncovers the highly-connected component that exhibits strong internal interactions \cite{Benson2016HigherorderOO,Spirin2003ProteinCA,liu2018graph}. 
Whereas, in the context of the real world, the discovery of a single CDS offers limited insights. 
Listing the top-$k$ CDSes is desired, but due to the substantial overlap inherent in $h$-cliques \cite{Wang2013RedundancyawareMC,Yuan2015DiversifiedTC}, the vanilla top-$k$ CDSes may refer to the same dense region, still providing limited structural insights. 

Therefore, detecting the top-$k$ non-overlapping, locally maximal, dense, and compact, i.e., \emph{locally $h$-clique densest subgraphs (L$h$CDS)} attracts great interest. For example, Figure \ref{fig:realexample} shows the relationships between a subset of characters in ``\textit{Harry Potter}''. The top-$1$ and top-$2$ L$3$CDSes are the blue and green subgraphs, respectively. The top-$1$ L$3$CDS is a family named \textit{Weasley}, and the top-$2$ L$3$CDS is an organization named \textit{Death Eaters}, which indicate the potential of L$h$CDS discovery for mining diverse dense communities.
\begin{figure}[h]
	\centering
	\includegraphics[width=0.9\linewidth]{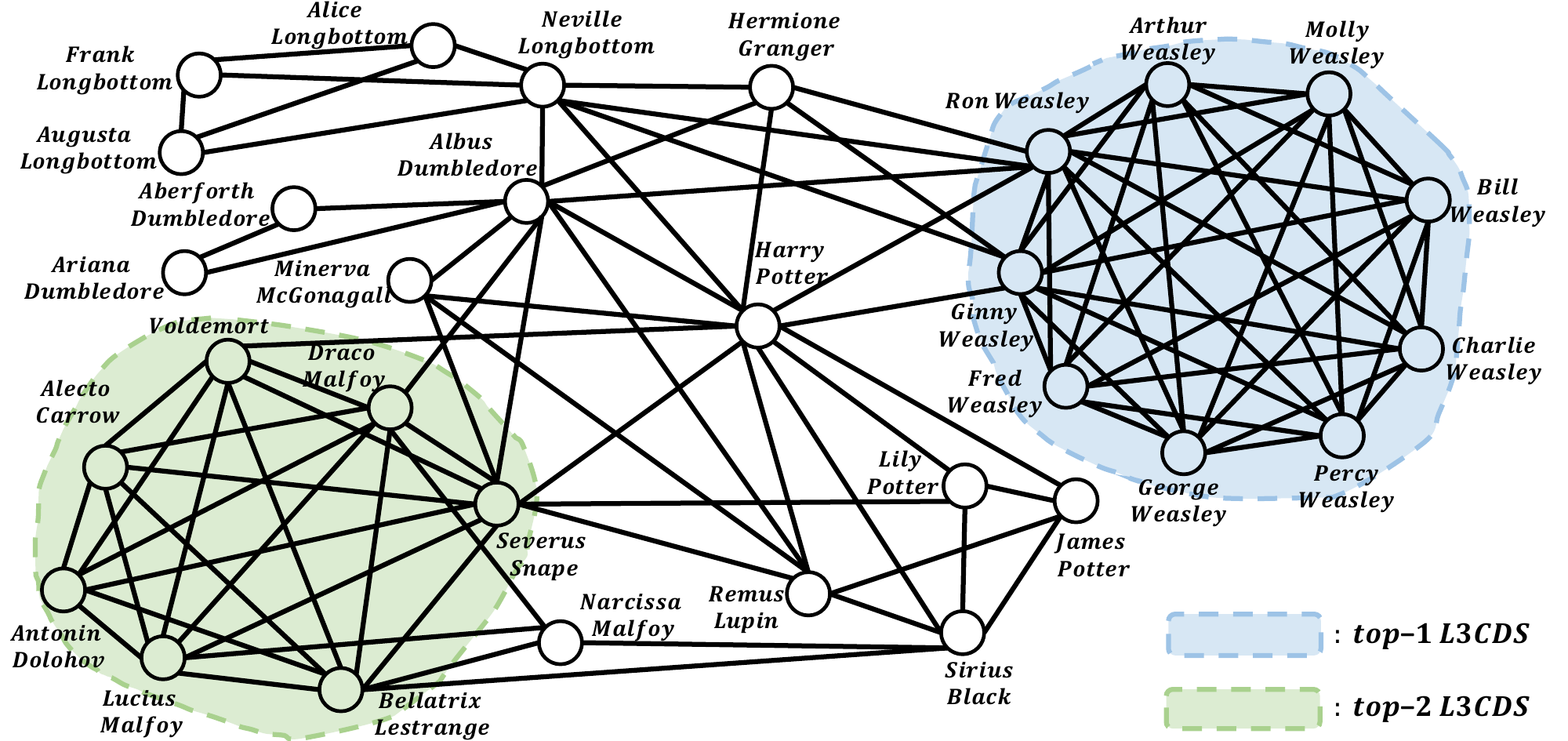}
	\caption{Part of the ``Harry Potter'' Network}
	\label{fig:realexample}
\end{figure}

However, no efficient and exact algorithm is known for detecting L$h$CDS yet. The closest work to L$h$CDS discovery is the locally densest subgraph (LDS) discovery \cite{qin2015locally}, which is a special case of L$h$CDS with $h=2$. But LDS only considers the density of edges, and it is challenging to generalize the techniques to arbitrary $h$. Firstly, the $h$-clique compactness is harder to evaluate, because the number of $h$-cliques can be several orders of magnitude more than the number of edges. Secondly, an $h$-clique spans over $h$ vertices, making the subgraph division more difficult. Thirdly, verification of L$h$CDS is more complex than verifying LDS since the clique density and clique compactness are harder to evaluate and verify. 

To address the above difficulties, we jointly consider the $h$-clique compact number estimation and L$h$CDS detection, 
so as to design a new \emph{iterative propose-prune-and-verify (\texttt{IPPV})} pipeline. 
\texttt{IPPV} is composed of the following iterative steps: (1) estimate bounds of $h$-clique compact number to propose L$h$CDS candidates efficiently; (2) use the bounds to prune vertices that are definitely not in any L$h$CDS to narrow the candidate set; (3) use verification algorithm to ensure the exactness of IPPV for extracting L$h$CDS; (4) remained candidates reenter the above steps until the top-$k$ L$h$CDSes are found.
To the best of our knowledge, this paper is the first to explore the L$h$CDS detection. The key contributions of \texttt{IPPV} are as follows: 

\begin{enumerate}[leftmargin=*]
	\item  
	 The initial $h$-clique compact number bounds are proposed based on the structures of graphs, and we show that a convex programming, which provides $h$-clique diminishingly dense decomposition, can be extended to tighten the bounds. 
	\item We propose a tentative graph decomposition method to deal with the case when a clique is spanning multiple subgraphs to generate correct decomposition proposals. 
	\item Efficient verification is the critical part since the verification is complex for verifying both the $h$-clique density and $h$-clique compactness. We propose a novel fast verification algorithm by carefully constructing a size-reduced flow network and using the maximum flow algorithm. We prove the correctness and efficiency of the proposed fast verification algorithm. 
	\item At last, we further extend the \emph{iterative propose-prune-and-verify} pipeline to detect locally general pattern densest subgraphs. More than six patterns are investigated, showing the potential of detecting locally more general pattern densest subgraphs. 
\end{enumerate}
We theoretically verify the exactness and efficiency of the proposed algorithm and conduct extensive experiments with different quality measures on large real datasets to evaluate the algorithm.

\section{related work}
\subsection{Densest Subgraph}  
The algorithms to the densest subgraph (DS) problem can be classified into two categories: exact algorithms and approximation algorithms. 
The DS problem can be solved in polynomial time by exact methods based on maximum flow, linear programming, or convex optimization. Picard et al. \cite{Picard1982ANF} and Goldberg \cite{Goldberg1984FindingAM} first introduced the maximum-flow-based exact algorithm for the densest subgraph problem. Charikar \cite{Charikar2000GreedyAA} proposed an LP-based exact algorithm for the DS problem. The convex-optimization-based exact algorithm is proposed by Danisch et al. \cite{Danisch2017LargeSD} and can handle graphs containing tens of billions of edges. Fang et al. \cite{Fang2019EfficientAF} improved the efficiency of the flow-based exact algorithm by locating the densest subgraph in a specific $k$-core. Exact algorithms cannot scale well to large graphs, so a large number of works on faster approximation algorithms for the DS problem are also presented. 

Charikar \cite{Charikar2000GreedyAA} proposed a 2-approximation algorithm for the DS problem, which is known as the greedy peeling algorithm. Proving that the $k_{max}$-core is a 2-approximation solution to the DS problem, Fang et al. \cite{Fang2019EfficientAF} improved the greedy peeling algorithm based on $k_{max}$-core. Inspired by the multiplicative weights update method, Boob et al. \cite{Boob2019FlowlessED} designed an iterative version of the greedy peeling algorithm. Based on the MapReduce model, Bahmani et al. \cite{Bahmani2012DensestSI} proposed an $2(1+\epsilon)$-approximation algorithm, where $\epsilon > 0$. Based on the dual of Charikar’s LP relaxation, Harb et al. \cite{Harb2022FasterAS} presented a new iterative algorithm for the DS problem. Chekuri et al. \cite{Chekuri2022DensestSS} proposed a flow-based approximation algorithm for the DS problem.\\ 
\indent The DS problem has various variants focusing on different aspects and different types of graphs. Two recent surveys \cite{lanciano2023survey,luo2023survey} detail different variations of the DS problem and their applications to different types of graphs, such as directed graphs \cite{Charikar2000GreedyAA}, labeled graphs \cite{fazzone2022discovering}, and uncertain graphs \cite{zou2013polynomial}.

\subsection{$h$-clique Densest Subgraph}
Tsourakakis \cite{tsourakakis2015k} defined the notion of $h$-clique density and introduced the $h$-clique densest subgraph (CDS) problem.
Mitzenmacher et al. \cite{Mitzenmacher2015ScalableLN} presented a sampling scheme called the densest subgraph sparsifier, yielding a randomized algorithm that produces a well-approximate solution to the CDS problem.
Fang et al. \cite{Fang2019EfficientAF} proposed more efficient exact and approximation algorithms for the CDS problem.
Sun et al. \cite{sun2020kclist++} aimed at developing near-optimal and exact algorithms for the CDS problem on large real-world graphs. They modified the Frank-Wolfe algorithm for CDS to their algorithm kClist++ and proved the effectiveness of the proposed algorithm. 

\subsection{Locally Densest Subgraph}
The locally densest subgraph (LDS) problem is a variant of the densest subgraph (DS) problem.
Qin et al. \cite{qin2015locally} proposed a method to discover the top-$k$ representative locally densest subgraphs of a graph. The method involves defining a parameter-free definition of an LDS, showing that the set of LDSes in a graph can be computed in polynomial time, and proposing three novel pruning strategies to reduce the search space of the algorithm.
Trung et al. \cite{Trung2023VerificationFreeAT} observed the hierarchical structure of maximal $\rho$-compact subgraphs and presented verification-free approaches to improve the efficiency of finding top-$k$ LDSes.
Ma et al. \cite{ma2022finding} proposed a convex-programming-based algorithm called LDScvx to the LDS problem by introducing the concept of the compact number and using the relations of compactness to the LDS problem and a specific convex program.
Capitalizing on previous results \cite{qin2015locally}, Samusevich et al. \cite{samusevich2016local} studied the local triangle densest subgraph (LTDS) problem, which extended the LDS model to triangle-based density. It's worth noting that, in essence, LDS is a special instance of L$h$CDS when $h=2$; LTDS is a special instance of L$h$CDS when $h=3$.

\section{Preliminaries}
Given an undirected graph $G=(V,E)$, we use $\psi_{h}(V_{\psi_{h}}, E_{\psi_{h}})$ to denote an $h$-clique with $|V_{\psi_{h}}|$ vertices and $|E_{\psi_{h}}|$ edges. $\Psi_{h}(G)$ is the collection of $h$-cliques of $G$. $d_{\psi_{h}}(G)$ denotes the $h$-clique density of $G$, $d_{\psi_{h}}(G)=\frac{|\Psi_{h}(G)|}{|V|}$, and  $deg_{G}(v,\psi_{h})$ is the $h$-clique degree of $v$, i.e., the number of $h$-cliques containing $v$. Given a subset $S\subseteq V$, $G[S]=(S,E(S))$ is the subgraph induced by $S$, and $E(S)=E(G)\cap(S\times S)$. Table \ref{tab:notation} summarizes the main notations used in this paper.
\begin{table}[htbp]\footnotesize
	\caption{MAIN NOTATIONS}
	\begin{tabular}{c|l}
		\toprule
		\textbf{Notation} & \textbf{Definition} \\ 
		\midrule
		$G=(V,E)$ &  a graph with vertex set $V$ and edge set $E$\\ 
		$n,m$ & $n=|V|, m=|E|$ \\ 
		$G[S]$ & the subgraph induced by $S$ \\ 
		$\Psi_{h}(G)$ & the collection of all $h$-cliques of $G$ \\ 
		$\psi_{h}(V_{\psi_{h}}, E_{\psi_{h}})$ & an $h$-clique ($V_{\psi_{h}}$ is vertex set, $E_{\psi_{h}}$ is edge set) \\ 
		$d_{\psi_{h}}(G)$ & the $h$-clique density of $G$, $d_{\psi_{h}}(G)=\frac{|\Psi_{h}(G)|}{|V|}$ \\ 
		$\phi_{h}(u)$ & the $h$-clique compact number of vertex $u$ \\ 
		$deg_{G}(v,\psi_{h})$ & the $h$-clique degree of vertex $v$ in $G$ \\ 
		$\overline{\phi}_{h}(u)$ & an upper bound of $\phi_{h}(u)$ in $G$ \\ 
		$\underline{\phi}_{h}(u)$ & a lower bound of $\phi_{h}(u)$ in $G$ \\
		$CP(G,h)$ & the convex programming of $G$ for $h$-clique densest subgraph\\ 
		$\alpha$ & the weights distributed from $h$-cliques to vertices \\ 
		$r$ & the weights received by each vertex \\ 
		\bottomrule
	\end{tabular}
	\label{tab:notation}
\end{table}

A densest subgraph in a local region not only means that such a subgraph is not included in any other denser subgraph, but also requires the inner density to be compact and evenly distributed. 
Qin et al. \cite{qin2015locally} proposed the concept of $\rho$-compact, which gives a reasonable definition of locally densest subgraph. A graph $G$ is $\rho$-compact when removing any subset $S$ from $G$ removes at least $\rho \times |S|$ edges. Considering the $h$-clique density in a graph, we define an $h$-clique $\rho$-compact graph as:

\begin{definition}[$h$-clique $\rho$-compact]
	A graph $G=(V,E)$ is $h$-clique $\rho$-compact if $G$ is connected, and removing any subset of vertices $S\subseteq V$ will result in the removal of at least $\rho\times|S|$ $h$-cliques in $G$, where $\rho$ is a non-negative real number.
	\label{def:hcc}
\end{definition}

If $G$ is $h$-clique $\rho$-compact, then the $h$-clique degree of each vertex in $G$ is at least $\lceil\rho\rceil$, because removing any vertex will remove at least $\rho$ $h$-cliques. 
Besides, the $h$-clique density of an $h$-clique $\rho$-compact graph is at least $\rho$. For any $\hat{\rho} > \rho$, an $h$-clique $\hat{\rho}$-compact graph is also an $h$-clique $\rho$-compact graph, so we define the $h$-clique compactness of a graph $G$ as the largest $\rho$ such that $G$ is $h$-clique $\rho$-compact. A subgraph $G[S]$ of $G$ is a maximal $h$-clique $\rho$-compact subgraph if none of the supergraphs of $G[S]$ is $h$-clique $\rho$-compact.
\begin{prop}
	If a graph $G$ has $h$-clique density $d_{\psi_{h}}(G)$, then the $h$-clique compactness of the graph is at most $d_{\psi_{h}}(G)$, i.e., $\rho \le d_{\psi_{h}}(G)$.
	\label{prop:densandcompa}
\end{prop}

\begin{proof}
	Suppose the compactness of $G$ is higher than $d_{\psi_{h}}(G)$, then removing all vertices in $G$ will result in the removal of more than $d_{\psi_{h}}(G)\times|V|$ $h$-cliques, which means that the $h$-clique density of $G$ must be higher than $d_{\psi_{h}}(G)$, and contradicts the fact that the $h$-clique density of $G$ is $d_{\psi_{h}}(G)$. 
\end{proof}

Proposition \ref{prop:densandcompa} clarifies that the $h$-clique compactness of a graph cannot be greater than $d_{\psi_{h}}(G)$. We are then interested in finding the locally $h$-clique dense and compact subgraph $G[S]$ in $G$. We formally define a locally $h$-clique densest subgraph as follows.

\begin{definition}[Locally $h$-clique densest subgraph (L$h$CDS)]
	A subgraph $G[S]$ of $G$ is a locally $h$-clique densest subgraph of $G$ if the following two conditions hold:
	\begin{enumerate}
		\item[1.] $G[S]$ is an $h$-clique $d_{\psi_{h}}(G[S])$-compact subgraph;
		\item[2.] There does not exist a supergraph $G[S']$ of $G[S]$ ($S'\supsetneq S$), such that $G[S']$ is also $h$-clique $d_{\psi_{h}}(G[S])$-compact.
	\end{enumerate}
	\label{def:lhcds}
\end{definition}

\begin{prop}[Disjoint property]
	Suppose $G[S]$ and $G[S']$ are two L$h$CDSes in $G$, we have $S\cap S'=\emptyset$.
	\label{prop:disjoint}
\end{prop}

\begin{proof}
	Without loss of generality, we assume $d_{\psi_{h}}(G[S]) \geq d_{\psi_{h}}(G[S'])$. We prove the proposition by contradiction. Suppose $S\cap S'\neq \emptyset$. According to the definition of L$h$CDS, $G[S]\nsubseteq G[S']$. Since $G[S]$ and $G[S']$ are two L$h$CDSes, the graph induced by $S\cup S'$ is a connected $h$-clique $d_{\psi_{h}}(G[S'])$ - compact graph which is larger than $S'$. This contradicts the fact that $G[S']$ is an L$h$CDS.
\end{proof}
Proposition \ref{prop:disjoint} proves that all L$h$CDSes in a graph $G$ are pairwise disjoint. Therefore, the number of L$h$CDSes of $G$ is bounded by $|V|$, and the L$h$CDSes can be used to identify all the non-overlapping $h$-clique dense regions of a graph.

Most applications in the real world usually require finding the top-$k$ dense regions of a graph \cite{qin2015locally}, so we focus on finding the top-$k$ L$h$CDSes with the largest densities. When $k$ is large enough, all L$h$CDSes can be found. We formulate the problem as follows.
\begin{definition} [Locally $h$-clique densest subgraph Problem (L$h$CDS Problem)]
	Given a graph $G$, an integer $h$, and an integer $k$, the locally $h$-clique densest subgraph problem is to compute the top-$k$ L$h$CDSes ranked by the $h$-clique density in $G$.
\end{definition}

Figure \ref{fig:example} shows an example of the L$h$CDS. We use $S_{1}$,$S_{2}$, and $S_{3}$ to represent $\{v_{12},...,v_{17}\}$, $\{v_{2},...,v_{6}\}$, and $\{v_{8},...,v_{11}\}$. When $h=3$, the top-$1$ L$3$CDS is $G[S_{1}]$, which has a 3-clique density of $\frac{13}{6}$, since there are thirteen 3-cliques in it. The top-$1$ and top-$2$ L$4$CDSes are $G[S_{2}]$ and $G[S_{1}]$. They both have a 4-clique density of $1$. 
\begin{figure}[htbp]
	\centering
	\includegraphics[width=0.8\linewidth]{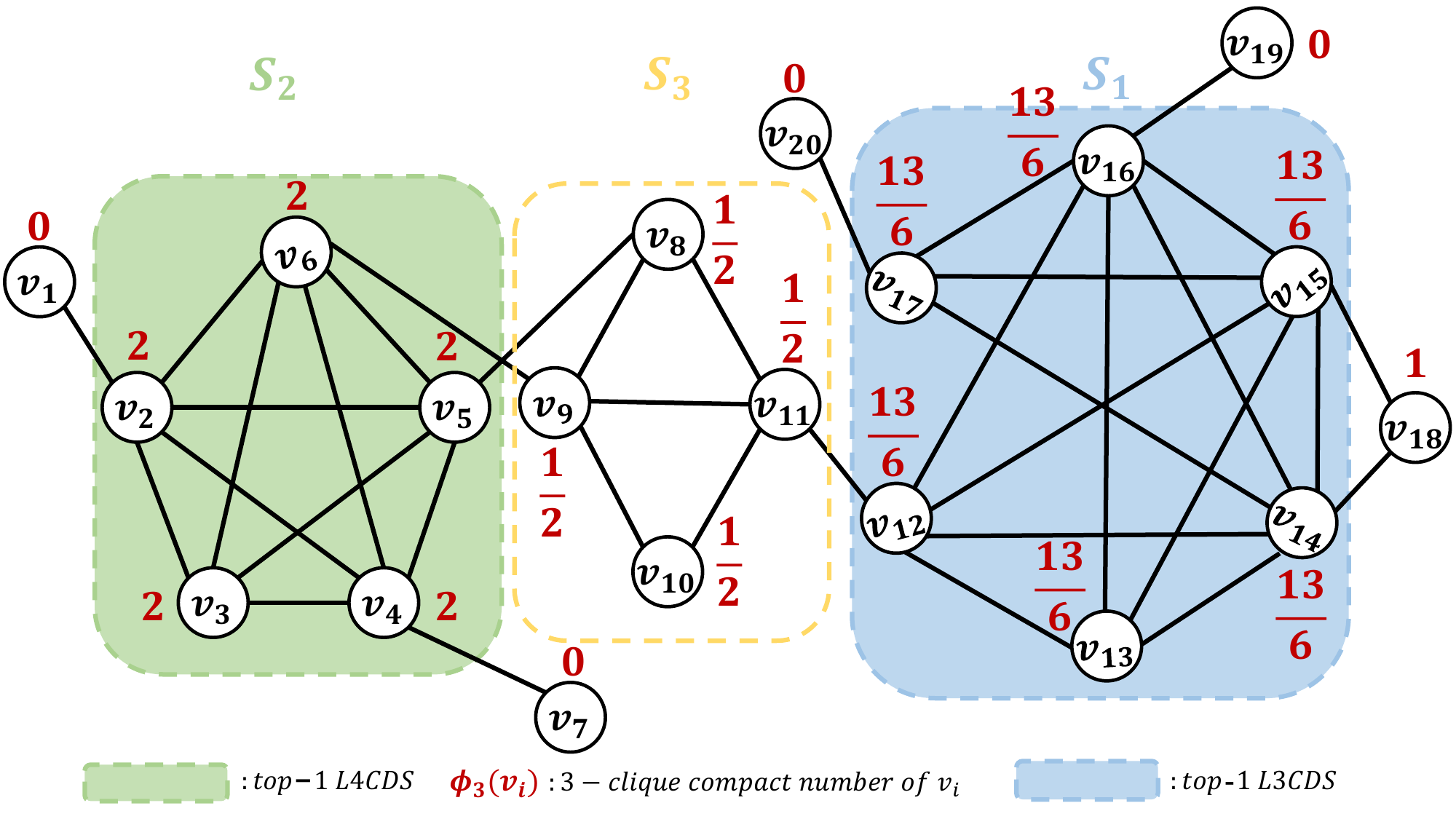}
	\caption{An example of the locally $h$-clique densest subgraph}
	\label{fig:example}
\end{figure}

Note that an edge in a graph $G$ is a 2-clique; therefore, the intensively studied LDS problem \cite{qin2015locally,ma2022finding} can be seen as an instance of the L$h$CDS problem when $h=2$.  Similarly, the LTDS problem \cite{samusevich2016local} is exactly the L$3$CDS problem. Therefore, the L$h$CDS problem studied in this paper provides a more general framework, and we boldly infer that our method can be generalized from $h$-clique to general patterns, which means that we can give an algorithmic framework to solve a wider range of locally pattern densest problems. 

\section{L$h$CDS Discovery}
In this section, we focus on the design of an L$h$CDS discovery algorithm. According to the concept of $h$-clique compactness, each subgraph of a graph $G$ has its own compactness. However, a vertex may be contained in various subgraphs with different compactness. Therefore, we introduce the concept of $h$-clique compact number for each vertex in a graph.  
\begin{definition}[$h$-clique compact number]
	Given a graph $G=(V,E)$, for each vertex $u\in V$, the $h$-clique compact number of $u$ is the largest $\rho$ such that $u$ is contained in an $h$-clique $\rho$-compact subgraph of $G$, denoted by $\phi_{h}(u)$.
	\label{def:hccn}
\end{definition} 
In the following theorem, we prove the relationship between the L$h$CDS and the $h$-clique compact numbers of vertices within it.
\begin{theorem}
	Given an  L$h$CDS $G[S]$ in $G$, for each vertex $u \in S$, the $h$-clique compact number of $u$ is equal to the $h$-clique density of $G[S]$, i.e., $\phi_{h}(u)=d_{\psi_{h}}(G[S])$.
	\label{theo:uin}
\end{theorem}
\begin{proof}
	As $G[S]$ is a maximal $h$-clique $d_{\psi_{h}}(G[S])$-compact subgraph, for each $u\in S$, there exists no other subgraph $G[S']$ containing $u$ such that $G[S']$ is an $h$-clique $\rho$-compact subgraph with $\rho > d_{\psi_{h}}(G[S])$. We prove the claim by contradiction. Suppose $G[S']$ is an $h$-clique $\rho$-compact subgraph with $\rho > d_{\psi_{h}}(G[S])$ and $u\in S'$, we have $d_{\psi_{h}}(G[S']) \geq \rho > d_{\psi_{h}}(G[S])$. First, $S'\subseteq S$, because $G[S]$ is a maximal $h$-clique $d_{\psi_{h}}(G[S])$-compact subgraph and $S'\cap S \neq \emptyset$. If we remove $U=S\backslash S'$ from $G[S]$, the number of $h$-cliques removed is $|\Psi_{h}(G[S])|-|\Psi_{h}(G[S'])| = d_{\psi_{h}}(G[S]) \times |S|-d_{\psi_{h}}(G[S'])\times |S'| < d_{\psi_{h}}(G[S]) \times (|S|-|S'|) = d_{\psi_{h}}(G[S]) \times |U|$. This contradicts that $G[S]$ is $h$-clique $d_{\psi_{h}}(G[S])$-compact. Hence, $d_{\psi_{h}}(G[S])$ is the $h$-clique compact number of all vertices in $S$.
\end{proof}

Based on Theorem \ref{theo:uin}, once we get the $h$-clique compact number of each vertex in $G$, we can arrange the vertices in descending order based on the $h$-clique compact number and then check whether the vertices with the same $h$-clique compact number satisfy the definition of L$h$CDS to obtain top-$k$ L$h$CDSes. For example, in Figure \ref{fig:example}, we list the $3$-clique compact numbers of all vertices of $G$. It is obvious that $G[S_{1}]$ and $G[S_{2}]$ are both L$3$CDSes. 

However, computing the $h$-clique compact numbers directly is difficult. So we jointly consider $h$-clique compact number and L$h$CDS to design a new ``iterative propose-prune-and-verify'' pipeline for top-$k$ L$h$CDS detection, which is called \texttt{IPPV}. In the proposal part, the true L$h$CDSes are allowed to be encapsulated in the proposed candidates, but without missing true L$h$CDSes. Proper graph decomposition methods should be designed, since a clique may span multiple subgraphs to be decomposed.  In the verification part, each correct L$h$CDS should be outputted, and L$h$CDS candidates that can be further pruned should be indicated.   

\begin{figure}[h]
	\centering
	\includegraphics[width=\linewidth]{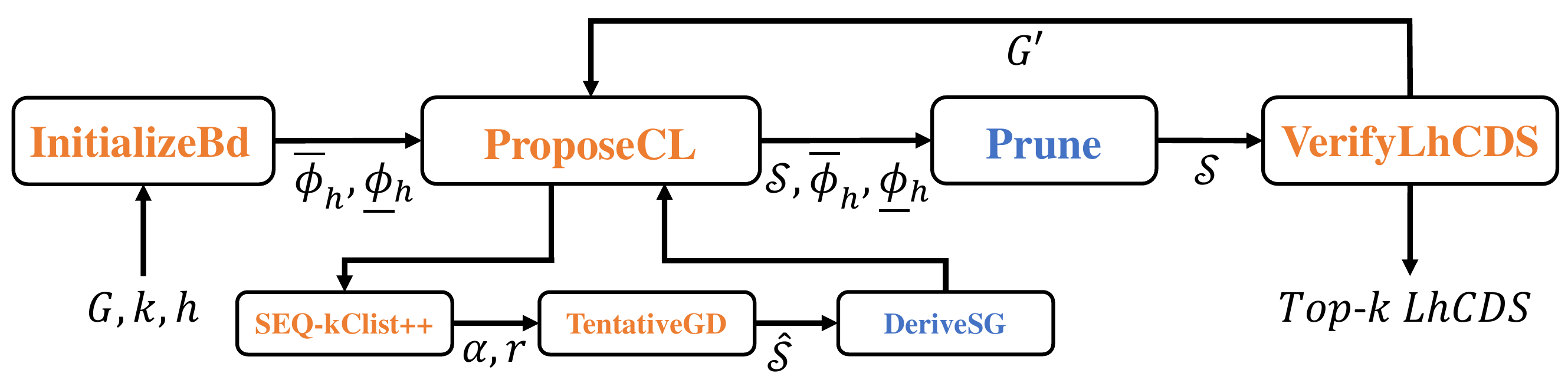}
	\caption{Flow diagram of \texttt{IPPV}}
	\label{fig:flowdiagram}
\end{figure}
Figure \ref{fig:flowdiagram} gives the flow diagram of \texttt{IPPV}. It has four main parts: \textbf{1)} calculate the initial bounds of the $h$-clique compact numbers of vertices; \textbf{2)} iteratively propose all L$h$CDS candidates (generating approximate $h$-clique compact numbers; decomposing the graph tentatively; grouping vertices and tightening bounds); \textbf{3)} prune invalid vertices; \textbf{4)} verify the locally densest property of all candidates to find top-$k$ L$h$CDSes, and we, in particular, propose a basic algorithm and a fast algorithm for verification.  As a general algorithm framework, all blue parts are extensions of existing methods, while all orange parts are our proof and innovation for this problem.

\subsection{Initial $h$-clique Compact Number Bounds}
In order to derive L$h$CDS candidates, we first give initial upper and lower bounds of $h$-clique compact number $\phi_{h}(u)$. 
Specifically, we denote $\overline{\phi}_{h}(u)$ and $\underline{\phi}_{h}(u)$ as the upper and lower bound of $\phi_{h}(u)$ in $G$. 
We use $(k,\psi_{h})$-core\cite{Fang2019EfficientAF}, which is a cohesive subgraph model, to compute the initial bounds. 

\begin{definition}[$(k,\psi_{h})$-core]
	Given a graph $G$, the $(k,\psi_{h})$-core is the largest subgraph of $G$, in which the $h$-clique degree of each vertex is at least $k$. The $h$-clique-core number of a vertex $u\in V$, denoted by $core_{G}(u,\psi_{h})$, is the highest $k$ of $(k,\psi_{h})$-core containing $u$.
\end{definition}
\begin{prop} $h$-clique compact number $\phi_{h}(u)$ has the following relations to the $h$-clique-core number $core_{G}(u,\psi_{h})$.
	\begin{enumerate}[leftmargin = \leftmargin-\widthof{nei}]
		\item A $(k,\psi_{h})$-core subgraph is $h$-clique $\frac{k}{h}$-compact. Thus, for any $u\in V$, $\underline{\phi}_{h}(u)$ can be assigned as $\frac{core_{G}(u,\psi_{h})}{h}$;
		\item If $G[S]$ is an L$h$CDS of $G$, for all $u\in S$, $core_{G}(u,\psi_{h}) \geq d_{\psi_{h}}(G[S])$. Thus, for any $u\in V$, $\overline{\phi}_{h}(u)$ can be assigned as $core_{G}(u,\psi_{h})$.
	\end{enumerate}
	\label{prop:corebound}
\end{prop}
\begin{proof}
	Any vertex in a $(k,\psi_{h})$-core subgraph is contained in at least $k$ $h$-cliques. By removal of any subset $S$ from the $(k,\psi_{h})$-core, at least $\frac{k}{h}\times |S|$ $h$-cliques would be removed. For any $u\in V$, there is an $h$-clique $\frac{core_{G}(u,\psi_{h})}{h}$-compact subgraph of $G$ that contains $u$, then $\frac{core_{G}(u,\psi_{h})}{h}$ is a lower bound of $\phi_{h}(u)$. The second relation can be obtained from the fact that an L$h$CDS $G[S]$ in a graph $G$ is a ($\lceil d_{\psi_{h}}(G[S])\rceil$,$\psi_{h}$)-core subgraph of $G$. For any $u\in V$, if an L$h$CDS contains $u$, then $core_{G}(u,\psi_{h})$ is an upper bound of $\phi_{h}(u)$.
\end{proof}
\begin{algorithm}\small
	\caption{Bound initialization: \texttt{InitializeBd}}
	\label{alg:initbound}
	\KwIn {$G=(V,E), h$} 
	\KwOut {$\overline{\phi}_{h}, \underline{\phi}_{h}$} 
	\lForEach{$u\in V$}{
		compute $core_{G}(u,\psi_{h})$
	}
	\ForEach{$u\in V$}
	{
		$\overline{\phi}_{h}(u)\leftarrow core_{G}(u,\psi_{h})$; $\underline{\phi}_{h}(u)\leftarrow \frac{core_{G}(u,\psi_{h})}{h}$;
	}
	\textbf{return} $\overline{\phi}_{h}$, $\underline{\phi}_{h}$;
\end{algorithm}

According to Proposition \ref{prop:corebound}, we can get the initial bounds of $h$-clique compact number $\phi_{h}(u)$ of $G$ (Lines 2-3) by Algorithm \ref{alg:initbound}.

\subsection{Candidate L$h$CDS Proposal} 
The initial upper and lower bounds for $h$-clique compact numbers from $h$-clique-core numbers are relatively loose. In this section, we focus on how to tighten the bounds and propose L$h$CDS candidates. 

\subsubsection{Overall Algorithm for Candidate L$h$CDS Proposal}
The overall candidate L$h$CDS proposal algorithm is given in Algorithm \ref{alg:dc}.
Approximate $h$-clique compact number is calculated via \texttt{SEQ-kClist++} (Line 1); the preliminary partition of $G$ and recalculated values are obtained via \texttt{TentativeGD} (Line 2); tighter upper and lower bounds for $h$-clique compact numbers and the further partition of $G$ (stable $h$-clique group) are calculated via \texttt{DeriveSG} (Line 3). The sub-procedures introduce each of the above functions (Lines 5-33).

\begin{algorithm}\small
	\caption{Candidate L$h$CDS proposal: \texttt{ProposeCL}}
	\label{alg:dc}
	\KwIn {$G=(V,E)$, number of iterations $T$, $\overline{\phi}_{h}, \underline{\phi}_{h}$} 
	\KwOut {$\mathcal{S}, \overline{\phi}_{h}, \underline{\phi}_{h}$} 
	($\alpha, r$)$\leftarrow$ \texttt{SEQ-kClist++} ($G,T$) ;\\
	$\hat{\mathcal{S}}, \alpha, r \leftarrow$ \texttt{TentativeGD} ($G,\alpha, r$);\\
	$\mathcal{S}, \overline{\phi}_{h}, \underline{\phi}_{h} \leftarrow$ \texttt{DeriveSG} ($\hat{\mathcal{S}}, \alpha, r, \overline{\phi}_{h}, \underline{\phi}_{h}$);\\
	return $\mathcal{S}, \overline{\phi}_{h}, \underline{\phi}_{h}$;\\
	\SetKwFunction{Fseq}{SEQ-kClist++}
	\SetKwProg{Fn}{Procedure}{}{}
	\Fn{\Fseq{$G,T$}}{
		\lForEach{ $h$-clique $\psi_{h}$ in $G$}
		{
			$\alpha_{u,\psi_{h}}\leftarrow \frac{1}{h}$ , $\forall u\in V_{\psi_{h}}$
		}
		\lForEach{ $u\in V$}
		{
			$r(u)\leftarrow \sum_{\psi_{h}\in \Psi_{h}(G): u\in \psi_{h}}\alpha_{u,\psi_{h}}$
		}
		\ForEach{ iteration $t = 1,...,T$}
		{
			$\gamma_{t}\leftarrow \frac{1}{t+1}$;
			$\alpha\leftarrow (1-\gamma_{t})*\alpha$; $r\leftarrow (1-\gamma_{t})*r$;\\
			\ForEach{ $h$-clique $\psi_{h}$}
			{
				$v_{min}\leftarrow arg min_{v\in \psi_{h}}r(v)$;\\
				$\alpha_{v_{min},\psi_{h}} \leftarrow \alpha_{v_{min},\psi_{h}} + \gamma_{t}$; $r(v_{min})\leftarrow r(v_{min}) + \gamma_{t}$;\\
			}  
		}
		\textbf{return} $(\alpha,r)$;
	}
	\SetKwFunction{Ften}{TentativeGD}
	\Fn{\Ften{$G,\alpha, r$}}{
		sort vertices in $V$ in descending order according to $r$;\\
		$P \leftarrow \{p|p = arg max_{p\leq q\leq n}d_{\psi_{h}}(G[V_{[1:q]}])\}$ ;\\
		$\hat{\mathcal{S}} \leftarrow$ partition $V$ according to $P$;\\
		\ForEach{$\psi_{h}\in \Psi_{h}(G)$}
		{
			$p \leftarrow \max \left\{1 \leq i \leq l: \psi_{h} \cap \hat{S}_{i} \neq \emptyset \right\}$; \\
			$s \leftarrow \sum_{u \in \psi_{h} \backslash \hat{S}_{p}} \alpha_{u,\psi_{h}}$; \\
			$\forall u \in \psi_{h} \backslash \hat{S}_{p}, \alpha_{u,\psi_{h}} \leftarrow 0$; \\
			$\forall u \in \psi_{h} \cap \hat{S}_{p}, \alpha_{u,\psi_{h}} \leftarrow \alpha_{u,\psi_{h}}+\frac{s}{\left|\psi_{h} \cap \hat{S}_p\right|}$;	
		}
		\lForEach{$u\in V$}
		{
			$r(u)\leftarrow \sum_{\psi_{h}\in \Psi_{h}(G): u\in \psi_{h}}\alpha_{u,\psi_{h}}$
		}
		\textbf{return} $\hat{\mathcal{S}}$, $\alpha$, $r$;
	}
	\SetKwFunction{Fder}{DeriveSG}
	\Fn{\Fder{$\hat{\mathcal{S}}, \alpha, r, \overline{\phi}_{h}, \underline{\phi}_{h}$}}{
		\While {$\hat{\mathcal{S}}$ is not empty}
		{
			$S' \leftarrow$ pop out the first candidate from $\hat{\mathcal{S}}$;
			$S \leftarrow S \cup S'$;\\
			\lIf {$S$ is a stable $h$-clique group}
			{
				put $S$ into $\mathcal{S}$; $S\leftarrow \emptyset$
			}
		}
		\ForEach{ $S\in \mathcal{S}$}
		{
			\ForEach{ $u\in S$}
			{
				$\overline{\phi}_{h}(u)\leftarrow min\{ \overline{\phi}_{h}(u), max_{v\in S} r(v)\}$;\\
				$\underline{\phi}_{h}(u)\leftarrow max\{ \underline{\phi}_{h}(u), min_{v\in S} r(v)\}$;
			}
		}
		\textbf{return} $\mathcal{S}$, $\overline{\phi}_{h}$, $\underline{\phi}_{h}$;
	}
\end{algorithm}

\subsubsection{Generate Approximate $h$-clique Compact Number}
Inspired by a classical convex programming \cite{Danisch2017LargeSD,ma2022finding}, we propose a convex programming for finding the diminishingly-$h$-clique-dense decomposition, and prove that the optimal solution of our convex programming is exactly the $h$-clique compact number of a graph $G$. 

Intuitively, the aim of CP($G,h$) is that each $h$-clique $\psi_{h}\in \Psi_{h}(G)$ tries to distribute its unit weight among its $h$ vertices such that the sum of the weight received by the vertices are as even as possible. We use $\alpha_{u,\psi_{h}}$ to represent the weight assigned to $u$ from $h$-clique $\psi_{h}$ and $r(u)$ to denote the sum of the weights assigned to $u$ from $h$-cliques that contain $u$. This intuition suggests that we can consider the objective function: $Q_{G, h}(\alpha):=\sum_{u \in V} r(u)^2$, in which $r(u)=\sum_{\psi_{h} \in \Psi_h(G): u \in \psi_{h}} \alpha_{u,\psi_{h}}$, for all $u \in V$. The convex programming is:
$$\mathrm{CP}(G, h):=\min \left\{Q_{G, h}(\alpha): \alpha \in \mathcal{D}(G, h)\right\},$$ 
where the domain is:
$$\mathcal{D}(G, h):=\left\{\alpha\in\prod_{\psi_{h} \in \Psi_h(G)} \mathbb{R}_{+}^{\psi_{h}}: \forall \psi_{h} \in \Psi_h(G), \sum_{u \in \psi_{h}} \alpha_{u,\psi_{h}}=1\right\}.$$

Here, we demonstrate that the $h$-clique compact numbers can be derived from the optimal solution of CP($G,h$).

\begin{theorem}
	Let $(\alpha^{*},r^{*})$ be an optimal solution to CP($G,h$). Then,  $\forall u\in V$, $\phi_{h}(u) = r^{*}(u)$, i.e., each $r^{*}(u)$ in $r^{*}$ is exactly the $h$-clique compact number of $u$.
	\label{theo:opsolution}
\end{theorem}

\begin{proof}
	For any vertex $u\in V$, let $S^{+}=\{ v\in V | r^{*}(v) > r^{*}(u)\}$, $S^{=}=\{ v\in V | r^{*}(v) = r^{*}(u)\}$, $S^{-}=\{ v\in V | r^{*}(v) < r^{*}(u)\}$, $u\in S^{=}$. $S^{+=}$ denotes the vertices that are contained by $h$-cliques that including vertices both in $S^{+}$ and $S^{=}$.
	We prove $G[S^{+}\cup S^{=}]$ is an $h$-clique $r^{*}(u)$-compact subgraph. 
	First, removing $S^{=}$ from $G[S^{+}\cup S^{=}]$ will result in the removal of $r^{*}(u)\times |S^{=}|$ cliques in $G[S^{+}\cup S^{=}]$. We know that for all $(v,w)\in E\cap (S^{+}\times S^{=})$, $r^{*}(v)>r^{*}(w)$ and $\alpha_{v,\psi_{h}(v,w\in \psi_{h})} = 0$. Otherwise, if there exists $ (v,w) \in E\cap (S^{+}\times S^{=})$ such that $\alpha_{v,\psi_{h}(v,w\in \psi_{h})}>0$, there exists $r^{*}(v)-r^{*}(w)>\epsilon>0$. We can reduce $\alpha_{v,\psi_{h}(v,w\in \psi_{h})}$ by $\epsilon$ and increase $\alpha_{w,\psi_{h}(v,w\in \psi_{h})}$ by $\epsilon$, and the objective function be decreased by $2\epsilon(r^{*}(v)-r^{*}(w)-\epsilon)$, which contradicts the optimality of $r^{*}$. Similarly, we can prove that for all $(v,w)\in E\cap (S^{=}\times S^{-})$, $r^{*}(v)>r^{*}(w)$ and $\alpha_{v,\psi_{h}(v,w\in \psi_{h})} = 0$. Therefore, $r^{*}(u)\times |S^{=}| = \sum_{\psi_{h}\in \Psi_{h}(G): v\in S^{=}, v\in \psi_{h}} \alpha_{v,\psi_{h}} = |\Psi_{h}(G[S^{=}])\cup \Psi_{h}(G[S^{+=}])|$. $r^{*}(u)\times |S^{=}|$ is exactly the number of $h$-cliques to be removed when removing $S^{=}$ from $G[S^{+}\cup S^{=}]$. 
	Meanwhile, for any $S'\subseteq S^{+}\cup S^{=}$,we have that $r^{*}(u)\times |S'| \leq \sum_{\psi_{h}\in \Psi_{h}(G): v\in S', v\in \psi_{h}} \alpha_{v,\psi_{h}} \leq \sum_{\psi_{h}\in \Psi_{h}(G[S^{+}\cup S^{=}]): v\in S', v\in \psi_{h}} 1$, which means removing any $S'\subseteq S^{+}\cup S^{=}$ from $G[S^{+}\cup S^{=}]$ will result in the removal of at least $r^{*}(u)\times |S'|$ $h$-cliques. Therefore, $G[S^{+}\cup S^{=}]$ is an $h$-clique $r^{*}(u)$-compact subgraph.
	Analogously, we can prove that for any other subset $S''$ containing $u$, $G[S'']$ is an $h$-clique $\rho$-compact subgraph, where $\rho\leq r^{*}(u)$, by contradiction. Therefore, $r^{*}(u)$ is the largest $\rho$ such that $u$ is contained in an $h$-clique $\rho$-compact subgraph of $G$, which is exactly the $h$-clique compact number of $u$.
\end{proof}
\begin{figure}[h]
	\centering
	\includegraphics[width=0.91\linewidth]{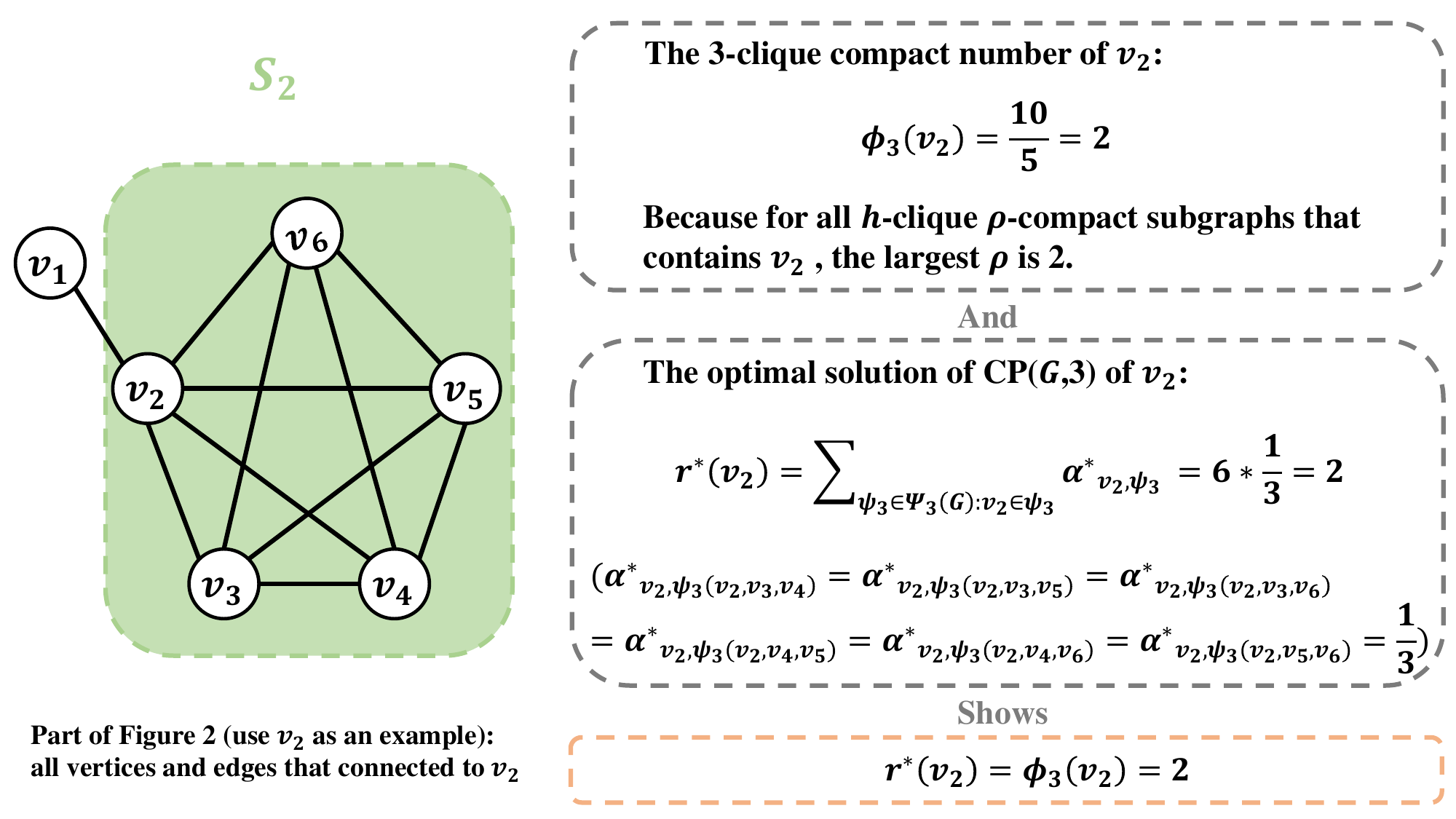}
	\caption{An example of the relationship between $r^{*}(u)$ and $\phi_{h}(u)$ of a vertex $u$ in $G$}
	\label{fig:partexample}
\end{figure}
Consider the convex programming CP($G,3$) for $G$ in Figure \ref{fig:example}, we use $v_{2}$ as an example, shown in Figure \ref{fig:partexample}. The $3$-clique compact number of $v_{2}$ is 2,  
and the optimal solution $r^{*}(v_{2})$ value is also 2. It is clear that for each $u\in V$, $r^{*}(u)$ is exactly $\phi_{h}(u)$.
\begin{theorem}
	The locally $h$-clique densest subgraph problem can be solved in polynomial time for any given positive integer $h$.
	\label{theo:polynomialtime}
\end{theorem}
\begin{proof}
	According to Theorem \ref{theo:opsolution}, the L$h$CDS problem is equivalent to CP($G,h$). CP($G,h$) is a convex quadratic programming problem, which has polynomial time solutions \cite{Goldfarb1991AnOP,Boyd2010ConvexO}. Therefore, the L$h$CDS problem is polynomial-time solvable. 
\end{proof}

However, CP($G,h$) is a problem with $|\Psi_{h}|$ equality constraints and $n\cdot|\Psi_{h}|$ variables. The time complexity of the classical polynomial-time solution of CP($G,h$) is $O((n\cdot|\Psi_{h}|)^{3}\cdot L)$ \cite{Goldfarb1991AnOP}, which requires long running time. Therefore, exactly attaining the $(\alpha^{*},r^{*})$ is difficult, so we use the approximate solution $(\alpha,r)$ of CP($G,h$) to tighten the $h$-clique compact bounds. Frank-Wolfe-based (FW-based) algorithm is efficient for finding approximate solutions to CP($G$) \cite{ma2022finding}. However, FW-based algorithm for $h$-clique densest requires a large amount of memory. \texttt{SEQ-kClist++}\cite{sun2020kclist++} is better for approximately calculating $\alpha_{u,\psi_{h}}$ for each $h$-clique $\psi_{h}$, $u\in \psi_{h}$, as well as $r(u)$ for each vertex $u$. All $\alpha_{u,\psi_{h}}$ are initialized to $\frac{1}{h}$ (Line 6). $r(u)$ stores the sum over all $\alpha_{u,\psi_{h}}$’s such that $\psi_{h}$ contains $u$ (Line 7). At each iteration, $\alpha$ and $r$ are modified simultaneously as follows. For each $h$-clique $\psi_{h}$, we find the minimum $r(v_{min})$ among $\psi_{h}$, and new values for the $\alpha_{v_{min},\psi_{h}}$ and the $r(v_{min})$ are computed as convex combinations (Lines 8-12).

\subsubsection{Tentative Graph Decomposition}
After getting approximate $(\alpha,r)$, we can derive a graph decomposition from the given ($\alpha,r$).
\begin{prop}
	Given an  L$h$CDS $G[S]$ in $G$, $\forall (u,v)\in E$, if $u\in S$ and $v\in V\backslash S$, we have $\phi_{h}(u)>\phi_{h}(v)$.
	\label{prop:vnotin}
\end{prop}

Considering vertices adjacent to $G[S_{1}]$ but not in $G[S_{1}]$ in Figure \ref{fig:example}, such as $v_{11}$ and $v_{18}$, their $3$-clique compact numbers fulfill Proposition \ref{prop:vnotin}: $\phi_{3}(v_{11}) = \frac{1}{2} < \frac{13}{6}$, $\phi_{3}(v_{18}) = 1 < \frac{13}{6}$. Proposition \ref{prop:vnotin} is helpful for choosing L$h$CDSes from all subgraphs. 

We then propose \texttt{TentativeGD} to generate tentative graph decomposition for proposing L$h$CDS. The vertices in $V$ are sorted based on $r$ values descendingly (Line 15). The initial partition $\hat{\mathcal{S}}$ of the graph is extracted based on the descending order (Lines 16-17). For each $\psi_{h}\in \Psi_{h}(G)$, if the clique $\psi_{h}$ is contained in multiple vertex sets, the vertex set with the largest set index will be recorded as $p$, and the $\alpha$ value of $\psi_{h}$ will be redistributed to vertices in $\hat{S}_{p}$ (Lines 18-22). In other words, for the convenience of partition, the $\alpha$ value of $\psi_{h}$ straddling multiple vertex sets is redistributed to a vertex set with the lowest $r$ value. Finally, the $r$ values of all vertices in $V$ are recalculated (Line 23).

\subsubsection{Stable $h$-clique Group Derivation}
After getting the initial bounds of $h$-clique compact numbers in \texttt{InitializeBd} and a preliminary partition of the graph in \texttt{TentativeGD}, we consider obtaining tighter bounds of $h$-clique compact numbers and a further partition of the graph, to calculate L$h$CDS candidates.
Inspired by two concepts, stable subset \cite{Danisch2017LargeSD} and stable group \cite{ma2022finding}, for solving the $h$-clique densest subgraph problem, we propose the definition of the stable $h$-clique group. 

\begin{definition}[stable $h$-clique group]
	Given a feasible solution $(\alpha,r)$ to CP($G,h$), a stable $h$-clique group with respect to $(\alpha,r)$ is a non-empty vertex group $S\in V$ satisfying the following conditions:
	\begin{enumerate}[leftmargin = \leftmargin - \widthof{non}]
		\item For any $v\in V\backslash S$, $r(v)>max_{u\in S}r(u)$ or $r(v)<min_{u\in S}r(u)$;
		\item For any $v\in V$, if $r(v)>max_{u\in S}r(u)$, $\forall \psi_{h} (u,v \in \psi_{h}), \alpha_{v,\psi_{h}}= 0$;
		\item For any $v\in V$, if $r(v)<min_{u\in S}r(u)$, $\forall \psi_{h} (u,v \in \psi_{h}), \alpha_{u,\psi_{h}}= 0$.
	\end{enumerate}
	\label{def:hcsg}
\end{definition}
\begin{figure}[h]
	\centering
	\includegraphics[width=0.8\linewidth]{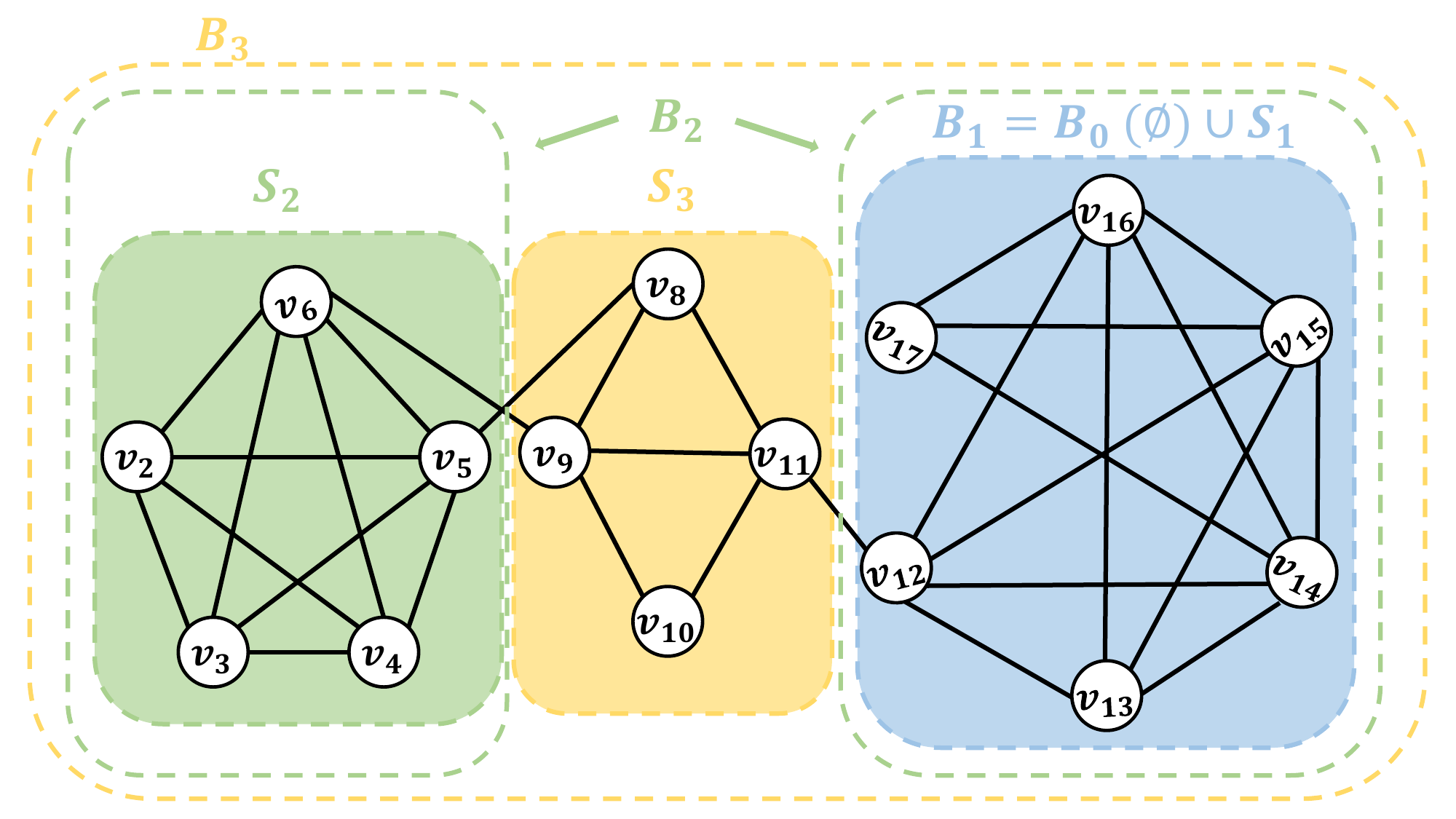}
	\caption{The relationship between stable $3$-clique subset $\mathcal{B}$ and stable $3$-clique group $\mathcal{S}$}
	\label{fig:relation}
\end{figure}
The concept of stable $h$-clique subset $\mathcal{B}$ is related to stable $h$-clique group $\mathcal{S}$, and the relationship between stable $h$-clique subset and stable $h$-clique group can be shown in Figure \ref{fig:relation} with $h=3$. All stable $h$-clique groups are disjoint, and a stable $h$-clique subset is the union of the previous stable $h$-clique subset and the first stable $h$-clique group outside this previous stable $h$-clique subset. Either $\mathcal{B}$ or $\mathcal{S}$ can form a consecutive subsequence of the whole sequence, and we only use the stable $h$-clique group in our algorithm.

\begin{theorem}
	Given a feasible solution $(\alpha,r)$ to CP($G,h$) and a stable $h$-clique group $S$ with respect to $(\alpha,r)$, for all $v \in S$, we have that $min_{u\in S}r(u)\leq \phi_{h}(v)\leq max_{u\in S}r(u)$.
	\label{theo:bound}
\end{theorem}

\begin{proof}
	According to Theorem \ref{theo:opsolution}, for all $u\in V, r^{*}(u) = \phi_{h}(u)$. Suppose there exists a vertex $v\in S$ such that $r^{*}(v) = \phi_{h}(v) < min_{u\in S}r(u) \leq r(v)$. Since $\sum_{u\in V}r(u)=\sum_{u\in V}r^{*}(u)$, correspondingly, there must exist another vertex $w\in V$, $r^{*}(w) = \phi_{h}(w)>r(w)$. The difference between $r(w)$ and $r^{*}(w)$ means that there exists $\psi_{h}$ containing both $v$ and $w$, $\alpha_{v,\psi_{h}}>0$. Since $S$ is a stable $h$-clique group, according to the third condition in Definition \ref{def:hcsg}, $r(w) > min_{u\in S}r(u)$. There exists $\epsilon > 0$, we can increase $r^{*}(v)$ by $\epsilon$ and decrease $r^{*}(w)$ by $\epsilon$ to decrease the value of the objective function. This contradicts that $r^{*}$ is the optimal solution to CP($G,h$). By the same token, for all $u\in V, r^{*}(u) = \phi_{h}(u)\leq max_{u\in S}r(u)$.
\end{proof}

Based on Theorem \ref{theo:bound}, the stable $h$-clique groups can give tighter bounds of $h$-clique compact numbers, so we propose \texttt{DeriveSG} algorithm to derive the stable $h$-clique groups, which are our L$h$CDS candidates. 
In \texttt{DeriveSG}, the subsets in $\hat{\mathcal{S}}$ are checked one by one; if the subset is a stable $h$-clique group, it will be pushed into the set of stable $h$-clique groups $\mathcal{S}$; otherwise, in the next iteration, the current subset $S$ will be merged with the next subset $S'$ (Lines 26–28). Then, the upper and lower bounds of $h$-clique compact numbers are updated based on Theorem \ref{theo:bound} (Lines 29–32).

\subsection{Pruning for Candidate L$h$CDS Derivation}
We prove that the following proposition can help to prune invalid vertices that are certainly not contained by any L$h$CDS.
\begin{prop}
	For any $v\in V$, $v$ is not contained by any L$h$CDS in $G$ if either of the following two conditions is satisfied.
	\begin{enumerate}[leftmargin = \leftmargin - \widthof{any}]
		\item If there exists $(u,v)\in E$, such that $\underline{\phi}_{h}(u)>\overline{\phi}_{h}(v)$, $v$ is invalid;
		\item Let $G'$ denote the graph after pruning all invalid vertices in condition (1). $\overline{\phi}_{h}^{G}(u)$ is the upper bound of $\phi_{h}^{G}(u)$ in $G$. For any $u$ in $G'$, if $\overline{\phi}_{h}^{G'}(v)<\underline{\phi}_{h}(v)$, $v$ is invalid.
	\end{enumerate}
	\label{prop:prunrule}
\end{prop}

\begin{proof}
	First, we prove condition (1). For any $u,v\in V$, $(u,v)\in E$, if $\underline{\phi}_{h}(u)>\overline{\phi}_{h}(v)$, then $\phi_{h}(u)>\phi_{h}(v)$. According to Proposition \ref{prop:vnotin}, $v$ is not contained in L$h$CDS, i.e., $v$ is invalid.
	
	For condition (2), $\overline{\phi}_{h}^{G'}(u)<\underline{\phi}_{h}(u)$ means that to form an $h$-clique $\underline{\phi}_{h}(u)$-compact subgraph containing $u$, some already pruned vertices are needed. So using the vertices in $G'$ only cannot form an $h$-clique $\underline{\phi}_{h}(u)$-compact subgraph containing $u$. Therefore, $u$ cannot be contained by any L$h$CDS in $G$, i.e., $v$ is invalid.  
\end{proof}

According to Proposition \ref{prop:prunrule}, we design Pruning Rule to prune invalid vertices by condition (1) and condition (2). 
An example can be seen in Figure \ref{fig:example} with $h=3$. $v_{9}$ and $v_{11}$ can be pruned, because for edge $(v_{6}, v_{9})$, $\underline{\phi}_{3}(v_{6}) = 2 > \overline{\phi}_{3}(v_{9}) = \frac{1}{2}$; for edge $(v_{11}, v_{12})$, $\underline{\phi}_{3}(v_{12}) = \frac{13}{6} > \overline{\phi}_{3}(v_{11}) = \frac{1}{2}$. Analogously, the vertices $v_{1}$, $v_{7}$, $v_{18}$, $v_{19}$, and $v_{20}$ are also pruned by condition (1). 

We denote the graph after pruning by $G'$. Some vertices in $G'$ become invalid vertices, because any L$h$CDS in $G$ containing these vertices needs to include some already pruned vertices, which can not form an L$h$CDS in $G'$. Therefore, we utilize condition (2). Based on Proposition \ref{prop:corebound}, for any vertex $u \in V(G')$, $core_{G}(u,\psi_{h})$ provides an upper bound of $\phi_{h}^{G'}(u)$. For example, after $v_{9}$ and $v_{11}$ are pruned, the upper bounds of $3$-clique compact numbers of $v_{8}$ and $v_{10}$ in graph $G'$ are $\overline{\phi}_{3}^{G'}(v_{8}) = \overline{\phi}_{3}^{G'}(v_{10}) = 0 < \underline{\phi}_{3}(v_{8}) = \underline{\phi}_{3}(v_{10}) = \frac{1}{2}$. So $v_{8}$ and $v_{10}$ are pruned using condition (2). 

We propose \texttt{Prune} algorithm shown in Algorithm \ref{alg:prune}. $G$ is replicated to $G'$ for pruning (Line 1). Condition (1) is used to remove invalid vertices in $G'$ (Lines 2-3); after computing the $h$-clique core numbers for all vertices in $G'$ (Line 4), condition (2) is applied to further remove invalid vertices in $G'$ (Lines 5-7). Finally, the L$h$CDS candidates are updated from the intersection of $h$-clique stable groups and the unpruned vertex sets (Line 8).

\begin{algorithm}\small
	\caption{Pruning algorithm: \texttt{Prune}}
	\label{alg:prune}
	\KwIn {$G=(V,E)$, $\mathcal{S}, \overline{\phi}_{h}, \underline{\phi}_{h}$} 
	\KwOut {$\mathcal{S}$} 
	$G'=(V(G'),E(G'))\leftarrow G$;\\
	\ForEach{$(u,v)\in E$}
	{
		\lIf {$\overline{\phi}_{h}(v) < \underline{\phi}_{h}(u)$}
		{
			remove $v$ from $G'$
		}
	}
	\lForEach{$u\in V(G')$}
	{
		compute $core_{G'}(u,\psi_{h})$
	}
	\While{there exists $u\in V(G'), core_{G'}(u,\psi_{h})<\underline{\phi}_{h}(u)$}
	{
		remove $u$ from $G'$;\\
		update $h$-clique-core numbers of vertices adjacent to $u$;
	}
	\lForEach{L$h$CDS candidate  $S\in \mathcal{S}$}
	{
		$S\leftarrow S\cap V(G')$
	}
	return $\mathcal{S}$;
\end{algorithm}

\subsection{L$h$CDS Verification}
Since candidate L$h$CDSes are obtained approximately, we need to confirm whether the candidates are L$h$CDSes.
\begin{prop} An L$h$CDS must satisfy the following properties: 
	\begin{enumerate}[leftmargin = \leftmargin-\widthof{nei}]
		\item Any subgraph of an L$h$CDS cannot be denser than itself;
		\item An L$h$CDS itself is compact, and any supergraph of an L$h$CDS cannot be more compact than itself.
	\end{enumerate}
	\label{prop:theproperty}
\end{prop}
\begin{proof}
	(2) of Proposition \ref{prop:theproperty} can be directly obtained from the definition of L$h$CDS. We prove (1) by contradiction. Suppose that there is a subgraph $G[S']$ in an L$h$CDS $G[S]$, $S'\subset S$, such that $d_{\psi_{h}}(G[S']) > d_{\psi_{h}}(G[S]) $. By removal of the set $U =S\backslash S'$from $G[S]$, we remove $|\Psi_{h}(G[S])|-|\Psi_{h}(G[S'])|$ $h$-cliques. Note that $|\Psi_{h}(G[S])|-|\Psi_{h}(G[S'])| = d_{\psi_{h}}(G[S])|S|-d_{\psi_{h}}(G[S'])|S'| < \linebreak d_{\psi_{h}}(G[S])(|S|-|S'|) = d_{\psi_{h}}(G[S])|U|$, which contradicts the fact that $G[S]$ is an L$h$CDS, i.e. $h$-clique $d_{\psi_{h}}(G[S])$-compact. 
\end{proof}

We need to verify: \textbf{1)} whether a candidate L$h$CDS $G[S]$ is self-densest and \textbf{2)} whether $G[S]$ is a maximal $h$-clique $d_{\psi_{h}}(G[S])$-compact subgraph in $G$. We use \texttt{IsDensest} \cite{sun2020kclist++} algorithm to check whether a candidate L$h$CDS $G[S]$ is self-densest. In this section, we focus on the verification of the second property, to verify whether $G[S]$ is a connected component of maximal $h$-clique $d_{\psi_{h}}(G[S])$-compact subgraphs in $G$.
We design a basic verification algorithm, and we further propose a fast algorithm by reducing the scale of the flow network. The correctness of both algorithms is proved. 

\subsubsection{Basic Verification Algorithm}
Given an L$h$CDS candidate $G[S]$, we propose an innovative flow network to derive maximal $h$-clique $d_{\psi_{h}}(G[S])$-compact subgraph $G'$ in $G$. If $G[S]$ is a connected component of $G'$, $G[S]$ is indeed maximal $h$-clique $d_{\psi_{h}}(G[S])$-compact subgraph and an L$h$CDS in $G$; otherwise, $G[S]$ is not an L$h$CDS. The flow network $\mathcal{F}(V_{\mathcal{F}}, E_{\mathcal{F}})$ is shown in Figure \ref{fig:flownetwork1}. The vertex set of $\mathcal{F}$ is $\{s\}\cup V\cup\Psi_{h}\cup\{t\}$. The arc set of $\mathcal{F}$ is given as follows. For each $h$-clique $\psi_{h}^{j}$, we add $h$ incoming arcs of capacity $1$ from the vertices which form $\psi_{h}^{j}$, and $h$ outgoing arcs of capacity of $h-1$ to the same set of vertices. For each vertex $v_{i}\in V$, we add an incoming arc of capacity $deg_{G}(v_{i},\psi_{h})$ from the source vertex $s$, and an outgoing arc of capacity $\rho*h$ to the sink vertex $t$. Given a parameter $\rho$, we prove that the flow network in \texttt{DeriveCompact} can be used to derive maximal $h$-clique $\rho$-compact subgraphs in $G$ according to Theorem \ref{theo:derivecompact}. 
\begin{figure}[htbp]
	\centering
	\includegraphics[width=0.9\linewidth]{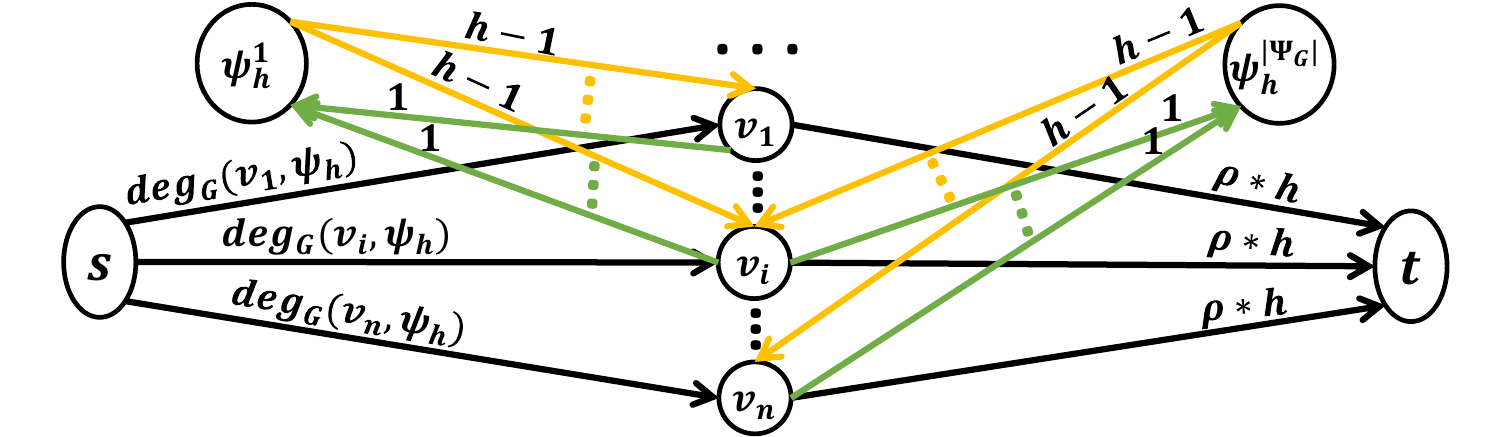}
	\caption{The flow network of \texttt{DeriveCompact}$(G,\rho,\emptyset)$}
	\label{fig:flownetwork1}
\end{figure}

\begin{theorem}
	If $G$ contains maximal $h$-clique $\rho$-compact subgraphs, then the result returned by DeriveCompact ($G,\rho - \frac{1}{|V|^{2}},\emptyset$) is the set of all maximal $h$-clique $\rho$-compact subgraphs in $G$.
	\label{theo:derivecompact} 
\end{theorem}
\begin{proof}
	Based on Proposition \ref{prop:disjoint}, two L$h$CDSes are disjoint. We use $G[S_{1}]$ to represent the union of all maximal $h$-clique $\rho$-compact subgraphs in $G$. $G[S_{2}]$ denotes the subgraph returned by DeriveCompact ($G,\rho - \frac{1}{|V|^{2}},\emptyset$), which is the largest subgraph in $G$ with maximum $|\Psi_{h}(G[S_{2}])|-\rho\times|S_{2}|$ \cite{Goldberg1984FindingAM}\cite{Fang2019EfficientAF}. We prove that $G[S_{1}]$ and $G[S_{2}]$ are the same. First, we prove that $G[S_{2}]$ is a subgraph of $G[S_{1}]$ by contradiction. Suppose a connected component $G[S]$ of $G[S_{2}]$ is not $h$-clique $\rho$-compact, then there exists a subset $S'\subseteq S$ such that removing $S'$ from $S$ will result in removing less $h$-cliques than $\rho\times |S'|$, then $|\Psi_{h}(G[S])|-|\Psi_{h}(G[S\backslash S'])|<\rho\times |S'|=\rho\times(|S|-|S\backslash S'|)$. We have $|\Psi_{h}(G[S])|-\rho\times|S|<|\Psi_{h}(G[S\backslash S'])|-\rho\times|S\backslash S'|$. Therefore, replacing $G[S]$ by its subgraph $G[S\backslash S']$ in $G[S_{2}]$ will enlarge the value of $|\Psi_{h}(G[S_{2}])|-\rho\times|S_{2}|$, which contradicts the condition that $G[S_{2}]$ has the maximum $|\Psi_{h}(G[S_{2}])|-\rho\times|S_{2}|$. Second, we prove that $G[S_{1}]$ is a subgraph of $G[S_{2}]$ by contradiction. Suppose $G[S_{1}]$ is not a subgraph of $G[S_{2}]$, according to the result before, we have $S_{2}\subset S_{1}$. There exists a subset $S\neq \emptyset$ and $S = S_{1}\backslash S_{2}$. Removing $S$ from $G[S_{1}]$ will result in removing at least $\rho\times |S|$ $h$-cliques, then $|\Psi_{h}(G[S_{1}])|-|\Psi_{h}(G[S_{2}])|\geq \rho\times|S|=\rho\times(|S_{1}|-|S_{2}|)$. We have $|\Psi_{h}(G[S_{1}])|-\rho\times|S_{1}|\geq |\Psi_{h}(G[S_{2}])|-\rho\times|S_{2}|$, so enlarging $G[S_{2}]$ to $G[S_{1}]$ will not decrease the value of $|\Psi_{h}(G[S_{2}])|-\rho\times|S_{2}|$, which contradicts the condition that $G[S_{2}]$ has the maximum $|\Psi_{h}(G[S_{2}])|-\rho\times|S_{2}|$. Therefore, the theorem is proved.
\end{proof}

\begin{algorithm}\small
	\caption{Basic L$h$CDS verification algorithm}
	\label{alg:compact}
	\KwIn {$G(V,E),S$} 
	\KwOut {VerifyL$h$CDS} 
	$\rho \leftarrow d_{\psi_{h}}(G[S])$, VerifyL$h$CDS $\leftarrow$ True;\\
	$G'\leftarrow$ \texttt{DeriveCompact} ($G,\rho-\frac{1}{|V|^{2}},\emptyset$);\\
	\textbf{return} $G[S]$ is a connected component in $G'$;\\
	\SetKwFunction{Fcomp}{DeriveCompact}
	\SetKwProg{Fn}{Procedure}{}{}
	\Fn{\Fcomp{$G,\rho,P$}}{
		$cnt\leftarrow 0$; $\Psi_{h}\leftarrow$  all the instances of $h$-clique $\psi_{h}$ in $G$;\\
		$V_{\mathcal{F}}\leftarrow \{s\}\cup V\cup\Psi_{h}\cup P\cup \{t\}$;\\
		\ForEach{$\psi_{h}\in \Psi_{h}$}
		{
			\ForEach{$v\in \psi_{h}$}
			{
				add an edge $\psi_{h}\rightarrow v$ with capacity $h-1$;\\
				add an edge $v\rightarrow \psi_{h}$ with capacity $1$;\\
			}
		}
		\ForEach{$\psi_{h}\in P$}
		{
			$cnt\leftarrow cnt$ of $\psi_{h}$;\\
			\ForEach{$v\in \psi_{h}$ and $v\in G$}
			{
				add an edge $\psi_{h}\rightarrow v$ with capacity $h-1$;\\
				add an edge $v\rightarrow \psi_{h}$ with capacity $1+\frac{h-cnt}{cnt}$;\\
				$deg_{G}(v,\psi_{h})\leftarrow deg_{G}(v,\psi_{h})+1+\frac{h-cnt}{cnt}$;
			}
		}
		\ForEach{$v\in V$}
		{
			add an edge $v\rightarrow t$ with capacity $\rho*h$;\\
			add an edge $s\rightarrow v$ with capacity $deg_{G}(v,\psi_{h})$;\\
		}
		Compute the minimum $s-t$ cut $(\mathcal{S},\mathcal{T})$ from the flow network $\mathcal{F}(V_{\mathcal{F}}, E_{\mathcal{F}})$;\\
		\textbf{return} $G[\mathcal{S}\setminus s]$;
	}	
\end{algorithm}

In Algorithm \ref{alg:compact}, we first derive all connected components of the $h$-clique $d_{\psi_{h}}(G[S])$-compact subgraph $G'$ in $G$ by \texttt{DeriveCompact} (Line 2). If $G[S]$ is a connected component of $G'$, the algorithm returns True (Line 3). In \texttt{DeriveCompact}, all the instances of $h$-clique is collected (Line 5). To build a flow network $\mathcal{F}(V_{\mathcal{F}}, E_{\mathcal{F}})$, a vertex set $V_{\mathcal{F}}$ is created, and vertices in $V_{\mathcal{F}}$ are linked by directed edges with different capacities (Lines 6-19). Then, the minimum $s-t$ cut $(\mathcal{S},\mathcal{T})$ is computed (Line 20). 

\subsubsection{Fast Verification Algorithm}
Although the basic verification algorithm can successfully verify whether a given subset is L$h$CDS, the scale of the flow network in algorithm \ref{alg:compact} is large, and the running time is long in large-scale graphs. We prove that the verification can be done by verifying only the subgraph $G[S]$ and the vertices around the subgraph $G[S]$, which is denoted by $G[T]$. Since $G[T]$ is much smaller than $G$, checking the minimum cut in $G[T]$ is much more efficient. Considering the complexity of the overlap of cliques, we propose a fast verification algorithm by constructing a smaller flow network based on $G[T]$. Based on the fact that only the $h$-cliques at the boundary of $G[T]$ affect $h$-clique compact numbers in $G[T]$ compared to the $h$-clique compact numbers in $G$, we use a set $P$ to record these $h$-cliques. For each $\psi_{h}^{P_{r}}\in P$, the number of vertices that are contained in both $\psi_{h}^{P_{r}}$ and $G[T]$ is $cnt_{P_{r}}$. The flow network $\mathcal{F}(V_{\mathcal{F}}, E_{\mathcal{F}})$ is shown in Figure \ref{fig:flownetwork2}. The vertex set of $\mathcal{F}$ is $\{s\}\cup V\cup\Psi_{h}\cup P\cup \{t\}$. We add the boundary $h$-clique set $P$ into $V_{\mathcal{F}}$ to ensure that the results of solving the flow network of $G[T]$ are precisely consistent with that of $G$. The arc set of $\mathcal{F}$ is given as follows. The arcs for vertices and $h$-cliques in $G[T]$ are the same as the former flow network. For each $\psi_{h}^{P_{r}}\in P$, we add $cnt_{P_{r}}$ incoming arcs of capacity $1+\frac{h-cnt_{P_{r}}}{cnt_{P_{r}}}$ from the vertices that both in $\psi_{h}^{P_{r}}$ and $G[T]$, and $cnt_{P_{r}}$ outgoing arcs of capacity of $h-1$ to the same set of vertices. 
\begin{figure}[htbp]
	\centering
	\includegraphics[width=0.9\linewidth]{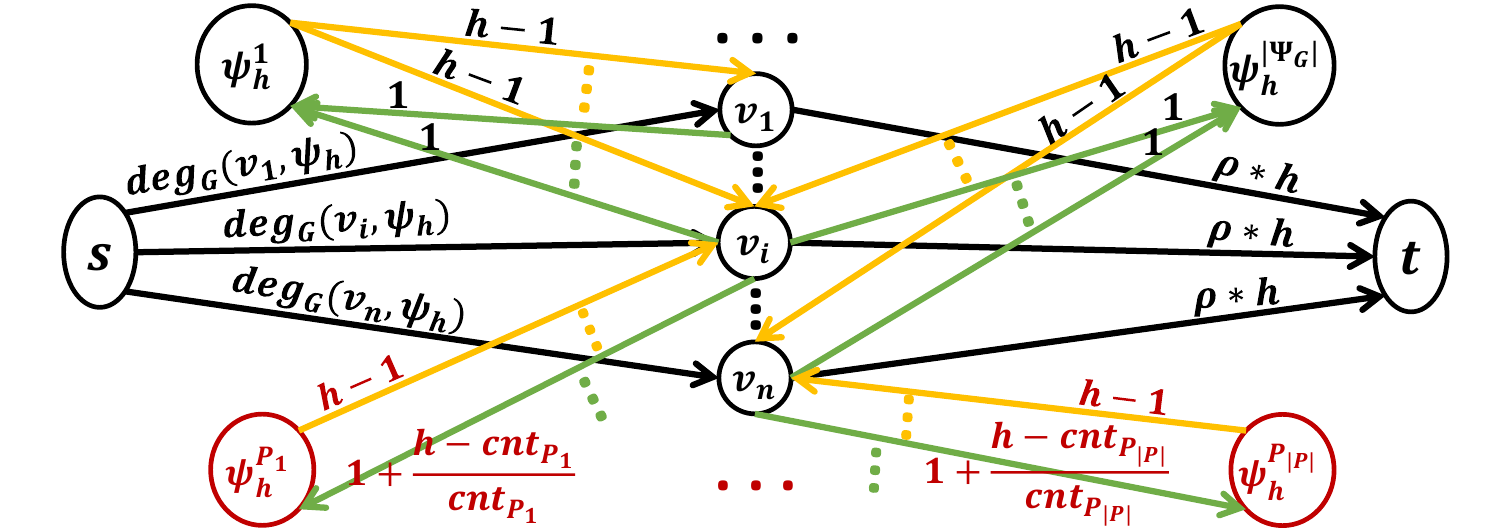}
	\caption{The flow network of \texttt{DeriveCompact}$(G,\rho,P)$}
	\label{fig:flownetwork2}
\end{figure}
\begin{algorithm}\small
	\caption{Fast L$h$CDS verification algorithm}
	\label{alg:verify}
	\KwIn {$G=(V,E)$, $S,\Psi_{h}(G), \overline{\phi}_{h}, \underline{\phi}_{h}$} 
	\KwOut {VerifyL$h$CDS} 
	$\rho \leftarrow d_{\psi_{h}}(G[S])$, VerifyL$h$CDS $\leftarrow$ True, Valid $\leftarrow$ True;\\
	$U\leftarrow$ an empty queue, $P\leftarrow \emptyset$, $T\leftarrow \emptyset$, $W\leftarrow \emptyset$, $cnt\leftarrow 0$;\\
	\ForEach{$u\in S$}
	{
		\lIf {$u \notin T$}
		{
			push $u$ to $U$, insert $u$ into $T$
		}
		\While {$U$ is not empty}
		{
			$v\leftarrow$ pop out the front vertex in $U$;\\
			\ForEach{$\psi_{h}\in\Psi_{h}(G)$ where $v\in \psi_{h}$}
			{
				Valid $\leftarrow$ True;\\
				\If{$\psi_{h}\notin W$}
				{
					\ForEach{$w\in \psi_{h}$}
					{
						\lIf{$\overline{\phi}_{h}(w)<\rho$}
						{
							Valid $\leftarrow$ False
						}
					}
					insert $\psi_{h}$ into $W$;
				}	
				\lElse
				{
					Valid $\leftarrow$ False
				}
				\If{Valid}
				{
					$cnt \leftarrow 1$;\\
					\ForEach{$w\in \psi_{h}$ and $w\neq v$}
					{
						\If{$w\notin T$ and $w$ is in any L$h$CDS}
						{
							VerifyL$h$CDS $\leftarrow$ False; 
						}
						\If{$\underline{\phi}_{h}(w)\leq\rho$}
						{
							\If{$w\notin T$}
							{
								push $w$ to $U$, insert $w$ into $T$;
							}
							$cnt \leftarrow cnt + 1$;
						}
					}
					\If{$cnt \neq h$ and $\psi_{h}\notin P$}
					{
						insert $\psi_{h}$ and $cnt$ into $P$;\\
						VerifyL$h$CDS $\leftarrow$ False; 
					}
				}
			}
			\ForEach{$(v,w)\in E$}
			{
				\If{$w\notin T$ and $\underline{\phi}_{h}(w)>\rho$}
				{
					VerifyL$h$CDS $\leftarrow$ False; 
				}
				\ElseIf{$w\notin T$ and $\overline{\phi}_{h}(w)>\rho$}
				{
					push $w$ to $U$, add $w$ into $T$;
				}
			}
		}
	}
	\lIf {VerifyL$h$CDS}
	{
		\textbf{return} True
	}
	$G'\leftarrow$ \texttt{DeriveCompact} ($G[T],\rho-\frac{1}{|V(G[T])|^{2}},P$);\\
	\textbf{return} $G[S]$ is a connected component in $G'$;
\end{algorithm}

In Algorithm \ref{alg:verify}, the $h$-clique density of $G[S]$ is assigned to $\rho$ (Line 1).  Then, a breadth-first search is performed. $U$ is used to store the vertices to be traversed (Line 4). The first vertex $v$ from $U$ is popped out (Line 6), and all $h$-cliques containing $v$ are iterated (Lines 7-25). For each $\psi_{h}$ that is not in $W$, if any vertex $w$ in $\psi_{h}$ satisfying $\overline{\phi}_{h}(w)<\rho$, $\psi_{h}$ will not affect the $h$-clique compact number of $w$ (Lines 9-13); if any vertex $w$ is in any outputted L$h$CDS, False is assigned to VerifyL$h$CDS (Lines 17-18); the number of vertices in $\psi_{h}$ satisfying $\underline{\phi}_{h}(w)\leq\rho$ is recorded and $\psi_{h}$ is added into $P$ (Lines 19-25). All neighbors of $v$ are iterated (Lines 26-30). 
For each neighbor $w$ that is not in $T$, if $\underline{\phi}_{h}(w)>\rho$, False is assigned to VerifyL$h$CDS (Line 28). 
If $\underline{\phi}_{h}(w) \leq \rho < \overline{\phi}_{h}(w)$, $w$ will be added into $U$ and $T$ (Line 30). 
If VerifyL$h$CDS is False, a subgraph $G[T]$ induced by $T$ and peripheral $h$-cliques in $P$ are used to compute all $h$-clique $\rho$-compact subgraphs in $G[T]$ via min-cut (Line 32). 
Finally, True is returned if $G[S]$ is maximal $h$-clique $\rho$-compact; otherwise, the algorithm returns False (Line 33). The
flow network here is much smaller.

\begin{theorem}
	Given a graph $G$ and a self-densest subgraph $G[S]$, $G[S]$ is an L$h$CDS of $G$ if and only if the fast L$h$CDS verification algorithm returns True.
	\label{theo:verify}
\end{theorem}

\begin{proof}
	On the one hand, if $G[S]$ is an L$h$CDS of $G$, $G[S]$ is still an L$h$CDS in $G[T]$, because only the $h$-cliques in $P$ might increase the $h$-clique compact numbers in $G[T]$ compared to the $h$-clique compact numbers in $G$. Otherwise, there exists a vertex $v$ with $h$-cliques in $P$ contained in the maximal $h$-clique $d_{\psi_{h}}(G[S])$-compact subgraph containing $G[S]$, and we can construct a larger $h$-clique $d_{\psi_{h}}(G[S])$-compact subgraph in $G$ by adding vertices with $\underline{\phi}_{h}(w) > d_{\psi_{h}}(G[S])$ connected to $v$, which contradicts that  $G[S]$ is an L$h$CDS. On the other hand, if $G[S]$ is not an L$h$CDS of $G$, we will find a larger $h$-clique $d_{\psi_{h}}(G[S])$-compact subgraph containing $G[S]$ in $G[T]$. Therefore, $G[S]$ is not an L$h$CDS in $G[T]$, and the algorithm returns False. Therefore, the algorithm returns True only when $G[S]$ is an L$h$CDS of $G$.
\end{proof}

%
%
%

\subsection{The L$h$CDS Discovery Algorithm (\texttt{IPPV})}
Combining all the algorithms above, we derive the L$h$CDS discovery algorithm, called the \texttt{IPPV} algorithm shown in Algorithm \ref{alg:LhCDS}. An empty stack $st$ is initialized, and $G'$ is assigned to $G$ (Line 1). The bounds of $h$-clique compact numbers are initialized via \texttt{InitializeBd} (Line 2). L$h$CDS candidates are derived via \texttt{ProposeCL} and \texttt{Prune} (Line 4-5). Next, the L$h$CDS candidates in $\mathcal{S}$ are reversely pushed into $st$, and the first L$h$CDS candidate in $st$, the one with the highest $\phi_{h}$ value, is popped out (Lines 6-7). The L$h$CDS candidate is verified by \texttt{IsDensest} (Line 8) and \texttt{VerifyL$h$CDS} (line 9). If $G[S]$ is an L$h$CDS, it will be outputted, and $k$ is decreased by 1 (Line 10). If $G[S]$ is not an L$h$CDS but is self-densest, $S$ is updated as the top L$h$CDS candidate from $st$ (Line 12). Then, $G[S]$ is assigned to $G'$ for the next iteration (Line 13). The above process is repeated until top-$k$ L$h$CDSes are found (line 3) or the stack is empty (Line 11). Our algorithms can also be extended to find all L$h$CDSes. 

\begin{algorithm}\small
	\caption{Iterative propose-prune-and-verify algorithm based on convex programming (\texttt{IPPV})}
	\label{alg:LhCDS}
	\KwIn {$G=(V,E)$, number of iterations $T$, an integer $k$} 
	\KwOut {top-$k$ L$h$CDSes} 
	$st\leftarrow$ an empty stack; $G'\leftarrow G$;\\
	$\overline{\phi}_{h}, \underline{\phi}_{h}\leftarrow$\texttt{InitializeBd}($G', h$);\\
	\While{$k>0$}
	{
		$\mathcal{S}, \overline{\phi}_{h}, \underline{\phi}_{h}\leftarrow$ \texttt{ProposeCL} ($G',T, \overline{\phi}_{h}, \underline{\phi}_{h}$) ;\\
		$\mathcal{S} \leftarrow$ \texttt{Prune} ($G', \mathcal{S}, \overline{\phi}_{h}, \underline{\phi}_{h}$);\\
		\lForEach{$S\in \mathcal{S}$ reversely}
		{
			push $S$ into $st$
		}
		$S\leftarrow$ pop out the top stable group from $st$ ;\\
		\If{\texttt{IsDensest} ($G[S]$)}
		{
			\If{\texttt{VerifyL$h$CDS} ($G,S,\overline{\phi}_{h}, \underline{\phi}_{h}$)}
			{
				output $G[S]$ ; $k\leftarrow k - 1$ ;
			}
			\lIf{$st$ is empty}
			{
				break
			}
			$S\leftarrow$ pop out the top stable group from $st$ ;\\
		}
		$G'\leftarrow G[S]$ ;\\
	}
\end{algorithm}
\begin{theorem}
	The L$h$CDS discovery algorithm \texttt{IPPV} is an exact algorithm, i.e., it can output all L$h$CDSes correctly.
	\label{theo:exactalgorithm}
\end{theorem}
\begin{proof}
	According to Theorem \ref{theo:bound}, the set of L$h$CDS candidates proposed by \texttt{IPPV} is a superset of all L$h$CDSes with tight bounds of $h$-clique compact number. Based on Proposition \ref{prop:prunrule}, the pruning part only prunes vertices that are not in any L$h$CDS and guarantees that the vertices in any L$h$CDS will not be pruned. Theorem \ref{theo:derivecompact} and Theorem \ref{theo:verify} state that the verification part verifies whether a candidate is exactly an L$h$CDS without misjudgment. And all candidates are repeated through the above bound updating, pruning, and verification operations until $k$ L$h$CDSes are outputted. Since the set of L$h$CDS candidates are arranged in descending order, the outputted $k$ L$h$CDSes must be the top-$k$ L$h$CDSes. Therefore, the proposed algorithm is an exact algorithm to output the top-$k$ L$h$CDSes. 
\end{proof}

\textbf{Complexity Analysis.} We use $T$ to denote the number of iterations that \texttt{SEQ-kClist++} needs. Each iteration of \texttt{SEQ-kClist++} costs $O(n+|\Psi_{h}|)$. We use $N_{CL}$ to represent the total number of L$h$CDS candidates, $N_{CL} \ll n$. Each iteration of verifying an L$h$CDS candidate costs $O(n+|\Psi_{h}|)$. $N_{Flow}$ is the number of times \texttt{IsDensest} and \texttt{VerifyL$h$CDS} are called. The time complexity of max-flow computation, which is $O((n+|\Psi_h|)^2\cdot (n+|\Psi_h|\cdot h))$ for \texttt{IsDensest} and \texttt{VerifyL$h$CDS} when Dinic Algorithm is Applied. The time complexity of \texttt{IPPV} is $O((T+N_{CL})\cdot(n+|\Psi_{h}|) + N_{Flow}\cdot(n+|\Psi_h|)^2\cdot(n+|\Psi_h|\cdot h))$. As $(T+N_{CL})\cdot(n+|\Psi_{h}|) \ll N_{Flow}\cdot(n+|\Psi_h|)^2\cdot(n+|\Psi_h|\cdot h)$, the time complexity of \texttt{IPPV} is $O(N_{Flow}\cdot(n+|\Psi_h|)^2\cdot(n+|\Psi_h|\cdot h))$. The memory complexity is $(n+|\Psi_{h}|)$. Note that the time complexity of the classical exact solution of only computing $h$-clique compact number is $O((n\cdot|\Psi_{h}|)^{3}\cdot L)$ \cite{Goldfarb1991AnOP}, which is much greater than that of our algorithm.

\section{L$hx$PDS Discovery}
A pattern (also known as a motif) \cite{Wuchty2003EvolutionaryCO,Leskovec2006PatternsOI,Hu2019DiscoveringMM} is a small connected subgraph that appears frequently in a larger graph, which can be considered as a basic module. Figure \ref{fig:pattern} shows all kinds of patterns with four vertices: $4a$-pattern, $\ldots$ ,$4f$-pattern.
\begin{figure}[h]
	\centering
	\includegraphics[width=\linewidth]{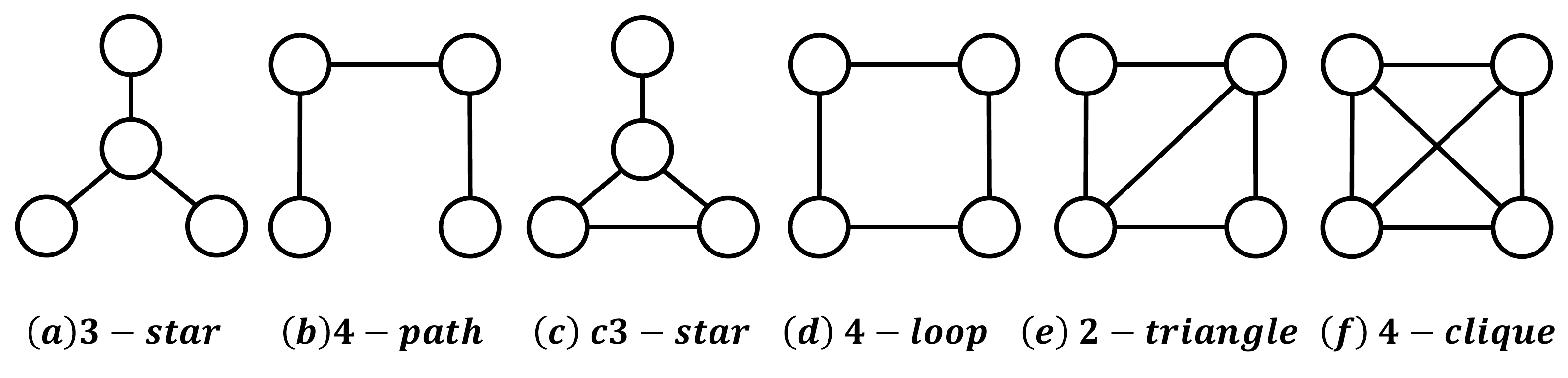}
	\caption{An example of all patterns with four vertices}
	\label{fig:pattern}
\end{figure}

We further show that the algorithm for the locally $h$-clique densest subgraph discovery problem can be extended to solve the locally general pattern densest subgraph discovery problem, which contributes to a deeper understanding of the organizational principles and functional modules within complex networks. 

\subsection{Densest Supermodular Set Decomposition}
In this section, we discuss the feasability of extending the $h$-clique problem to a general pattern problem. The convex programming of $h$-clique can be further generalized to the convex programming of supermodular sets, so that the convex programming for the general pattern densest subgraph problem and the corresponding compact number can be derived. A function $f : 2^{V} \rightarrow \mathbb{R}_{+}$ is said to be supermodular if $\forall A,B\subseteq V, f(A) + f(B) \leq f(A \cup B) + f(A\cap B)$. Harb et al. \cite{Harb2022FasterAS} proposed the densest supermodular subset (DSS) problem: given a normalized, nonnegative monotone supermodular function $f : 2^{V} \rightarrow \mathbb{R}_{+}$, return $S\subseteq V$ that maximizes $\frac{f(S)}{|S|}$. According to our observation, when $f(S)=|E(S)|$ and $f(S)=|\Psi_{h}(G[S])|$, the DSS problem is the DS and CDS problem, respectively. When $f(S)$ represents the number of a particular pattern in a graph, the problem is the densest problem of the proposed pattern. The convex program\cite{Harb2022FasterAS} for the densest supermodular set decomposition is $\mathrm{CP}(G):=\min \left\{\sum_{u \in V} r(u)^2\right\},$
subject to:
$r\in \left\{x\in \mathbb{R}^{V}|x\geq 0, x(S)\geq f(S) \text{ for all } S\subseteq V, x(V)=f(V) \right\}$.

With supermodularity, there is a property that each graph has a unique nested diminishingly decomposition for each type of density. The analysis of the generalization of CDS problem to DSS problem has triggered our thinking on the solution of locally general pattern densest problem.

\subsection{Locally General Pattern Densest Subgraph Problem}
Given an undirected graph $G=(V,E)$, $\psi_{hx}(V_{\psi_{hx}}, E_{\psi_{hx}})$ denotes a particular kind of pattern $x$ with $h$ vertices and $\Psi_{hx}(G)$ is the collection of the $hx$-patterns of $G$. $d_{\psi_{hx}}(G)=\frac{|\Psi_{hx}(G)|}{|V|}$ denotes the $hx$-pattern density of $G$. $deg_{G}(v,\psi_{hx})$ is the $hx$-pattern degree of $v$, i.e., the number of $hx$-patterns containing $v$. A graph $G=(V,E)$ is $hx$-pattern $\rho$-compact if $G$ is connected, and removing any subset of vertices $S\subseteq V$ will result in the removal of at least $\rho\times|S|$ $hx$-patterns in $G$. We can formally define a locally $hx$-pattern densest subgraph as follows.

\begin{definition}[Locally $hx$-pattern densest subgraph (L$hx$PDS)]
	A subgraph $G[S]$ of $G$ is an L$hx$PDS of $G$ if $G[S]$ is $hx$-pattern $d_{\psi_{hx}}(G)$-compact, and there does not exist a supergraph $G[S']$ of $G[S]$ ($S'\supsetneq S$), such that $G[S']$ is also $hx$-pattern $d_{\psi_{hx}}(G)$-compact.
	\label{def:lhxpds}
\end{definition}

Similarly, we formulate the locally $hx$-pattern densest subgraph problem as follows.

\begin{definition} [Locally $hx$-pattern densest subgraph Problem (L$hx$PDS Problem)]
	Given a graph $G$, an integer $h$, a pattern $x$ and an integer $k$, the L$hx$PDS problem is to compute the top-$k$ L$hx$PDSes ranked by the $hx$-pattern density in $G$.
\end{definition}

\begin{algorithm}\small
	\caption{The \texttt{IPPV} algorithm for L$hx$PDS}
	\label{alg:pattern}
	\KwIn {$G=(V,E)$, number of iterations $T$, an integer $k$} 
	\KwOut {top-$k$ L$hx$PDS} 
	$st\leftarrow$ an empty stack; $G'\leftarrow G$;\\
	$\overline{\phi}_{hx}, \underline{\phi}_{hx}\leftarrow$ \texttt{InitializeBd} for $hx$-pattern ($G', h, x$);\\
	\While{$k>0$}
	{
		$\mathcal{S}, \overline{\phi}_{hx}, \underline{\phi}_{hx}\leftarrow$ \texttt{ProposeCL} for $hx$-pattern ($G',T, \overline{\phi}_{hx}, \underline{\phi}_{hx}$);\\
		$\mathcal{S} \leftarrow$ \texttt{Prune} for $hx$-pattern ($G', \mathcal{S}, \overline{\phi}_{h}, \underline{\phi}_{h}$);\\
		\ForEach{$S\in \mathcal{S}$ reversely}
		{
			push $S$ into $st$;\\
		}
		$S\leftarrow$ pop out the top stable group from $st$;\\
		\If{\texttt{IsDensest} ($G[S]$)for $hx$-pattern}
		{
			\If{\texttt{VerifyL$hx$PDS} ($G,S,\overline{\phi}_{hx}, \underline{\phi}_{hx}$)}
			{
				output $G[S]$; $k\leftarrow k - 1$;\\
			}
			\If{$st$ is empty}
			{
				break;\\
			}
			$S\leftarrow$ pop out the top stable group from $st$;\\
		}
		$G'\leftarrow G[S]$;\\
	}
\end{algorithm}

\begin{figure*}[htbp]
	\centering
	
	\subfigure[{fb-pages-company}]{\includegraphics[width=0.24\textwidth]{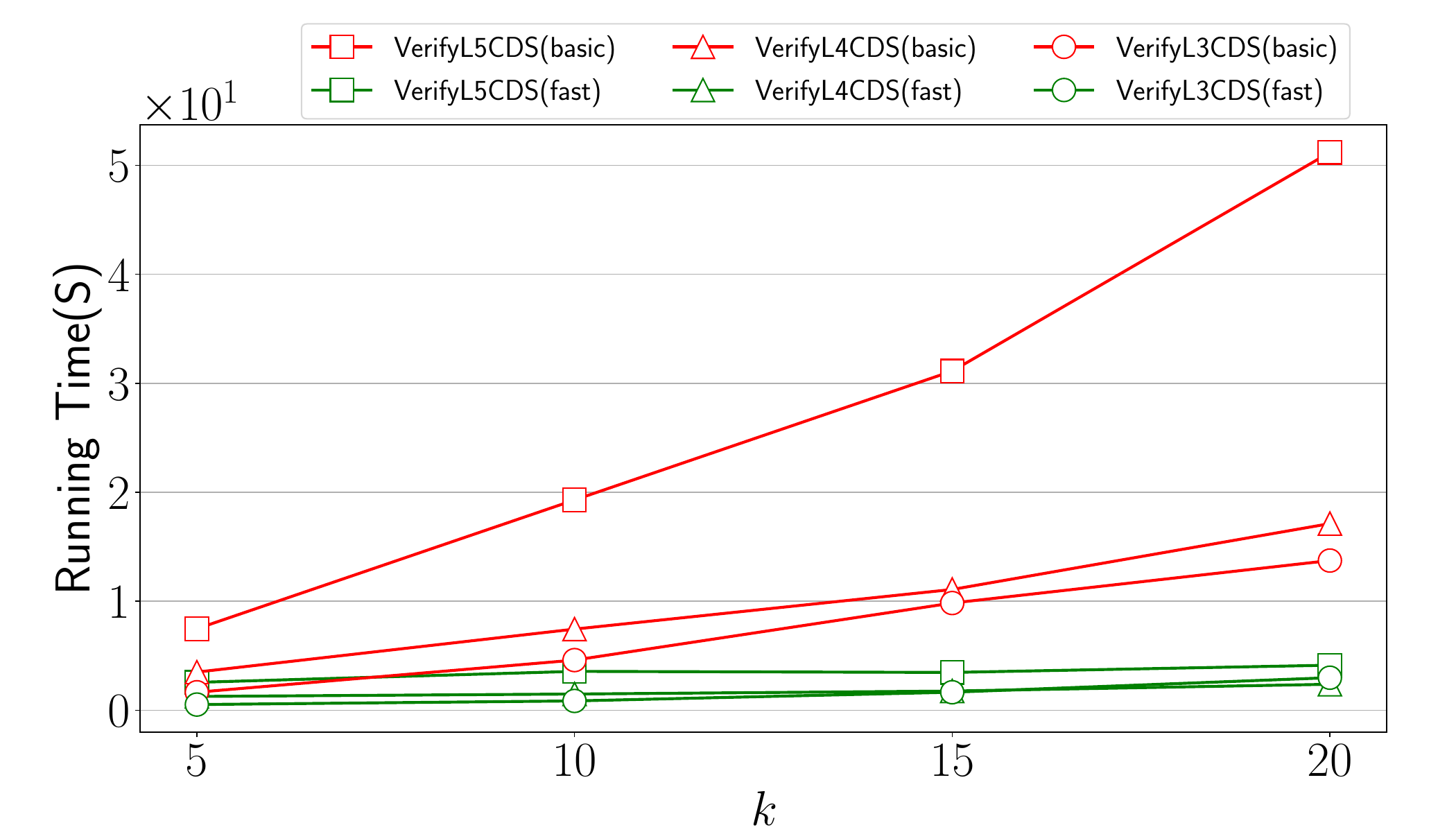}}
	\subfigure[{soc-hamsterster}]{\includegraphics[width=0.24\textwidth]{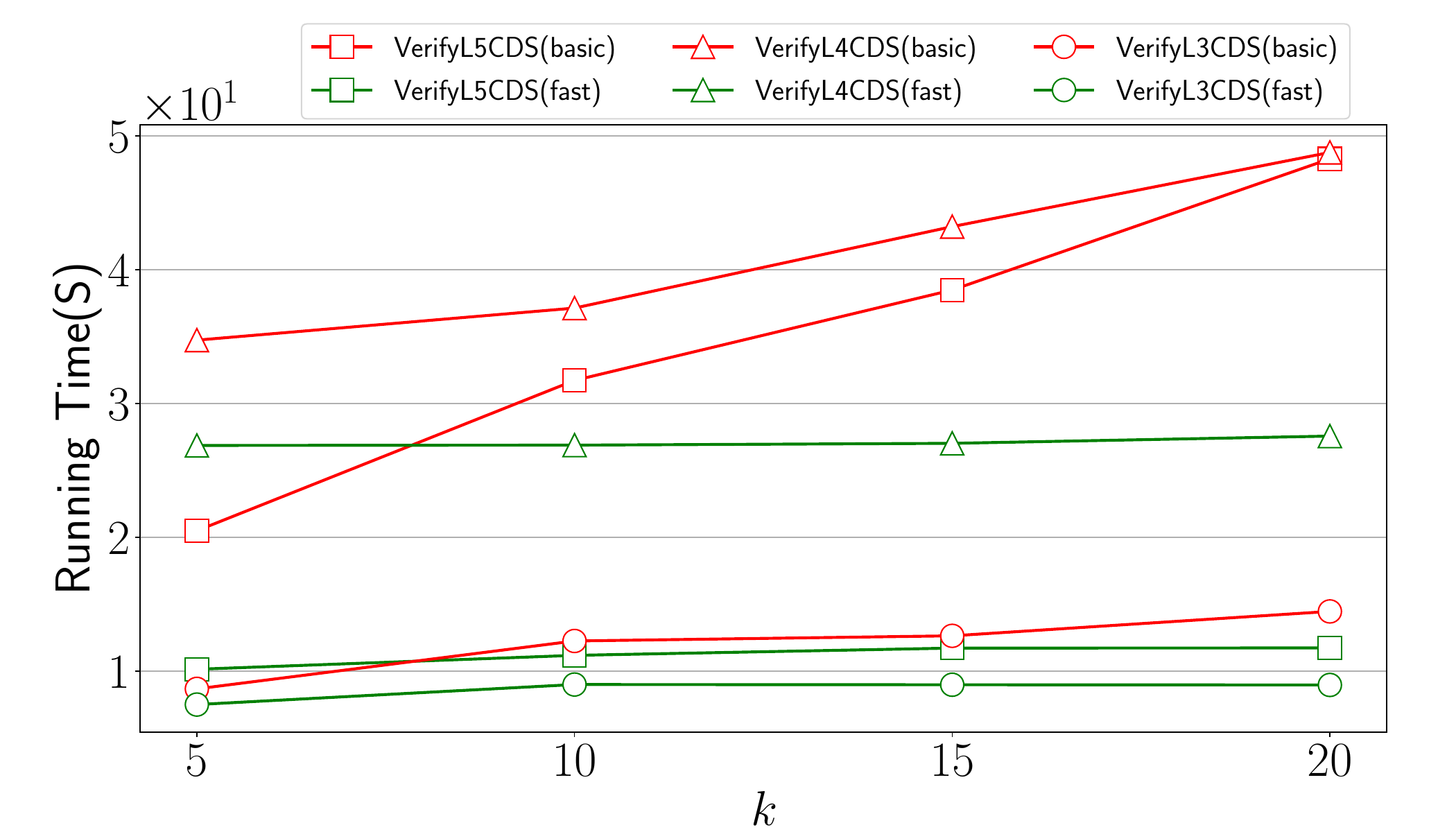}}
	\subfigure[{soc-epinions}]{\includegraphics[width=0.24\textwidth]{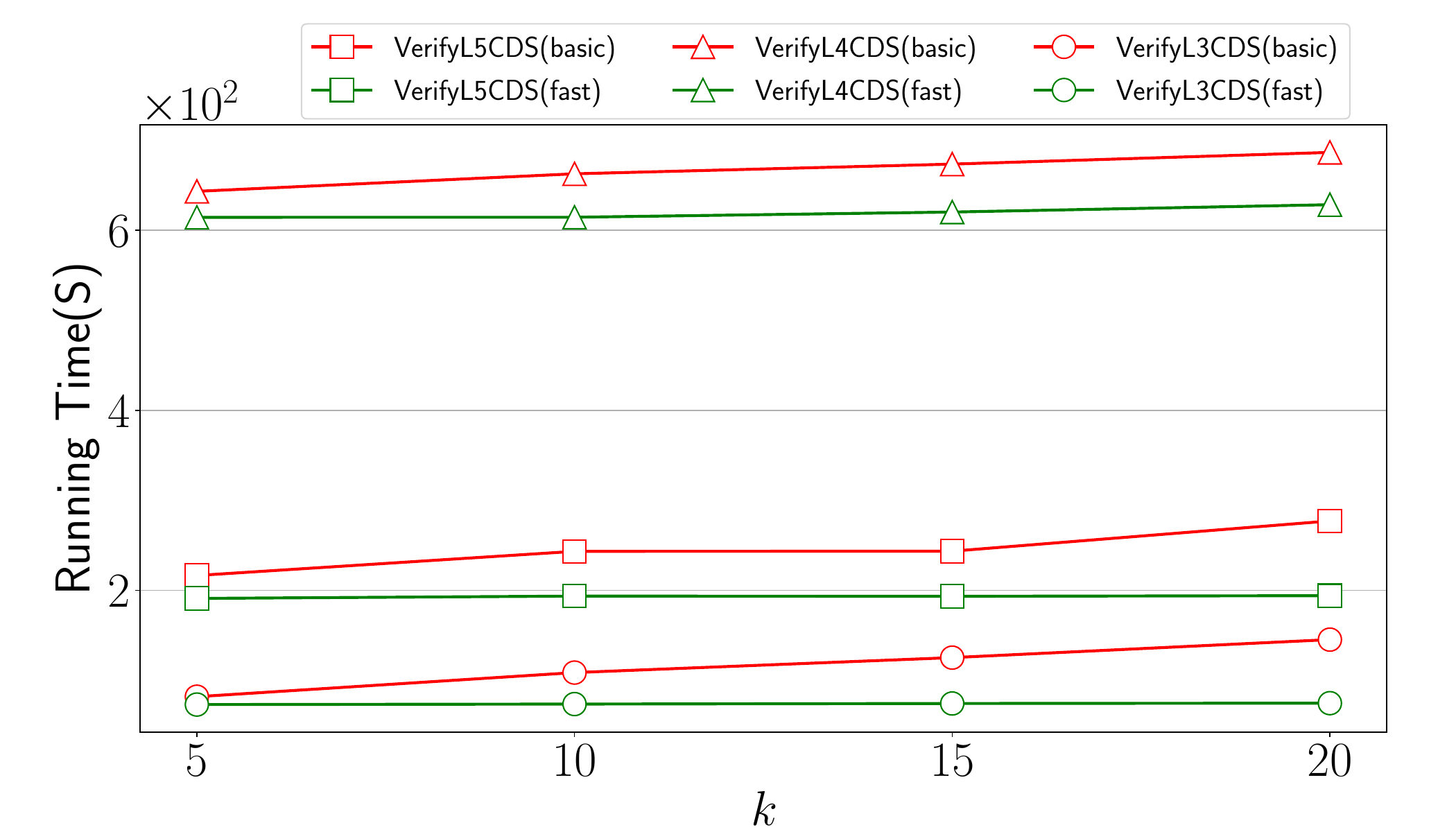}}
	\subfigure[{Email-Enron}]{\includegraphics[width=0.24\textwidth]{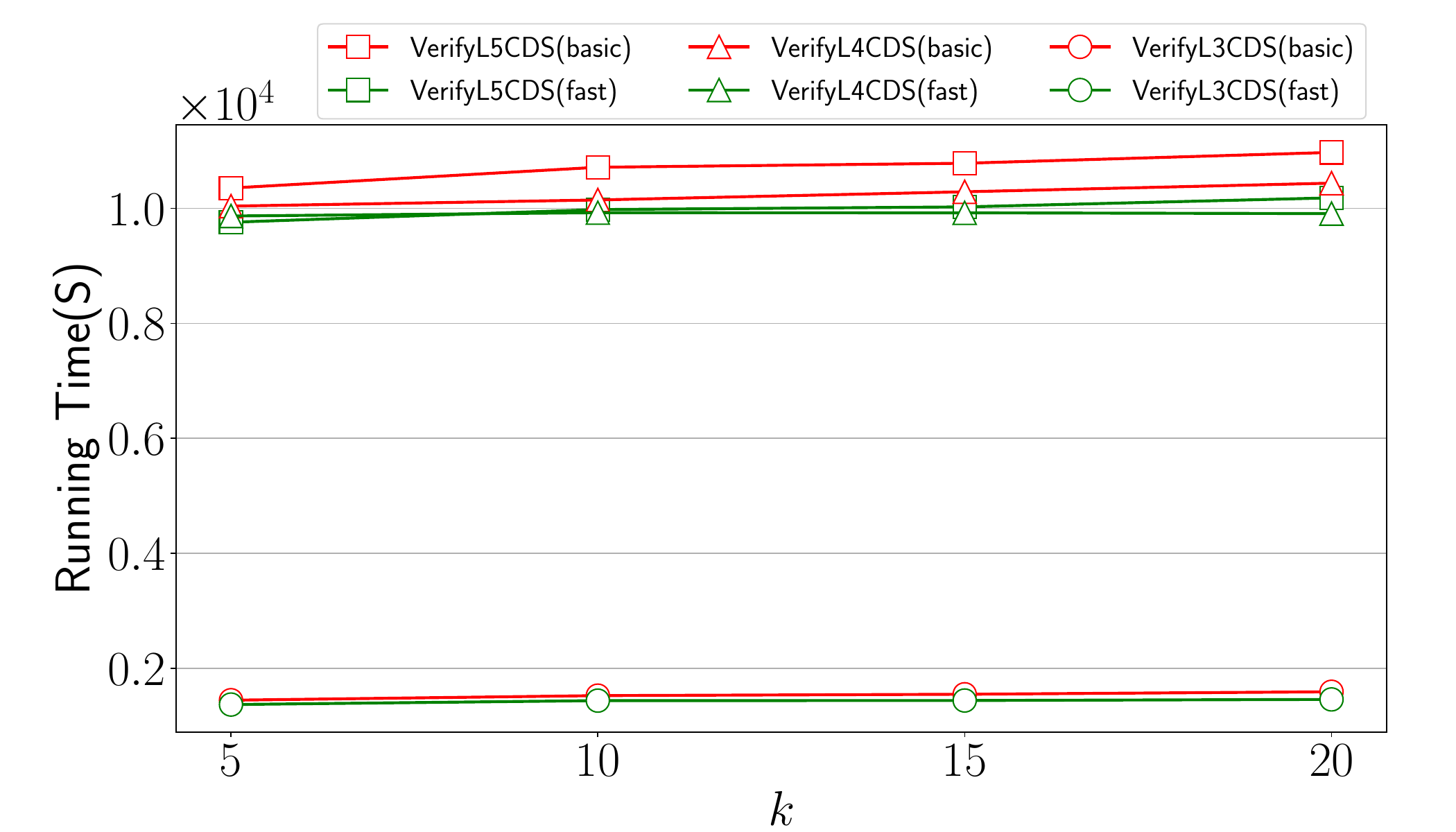}}
	
	\subfigure[{loc-gowalla}]{\includegraphics[width=0.24\textwidth]{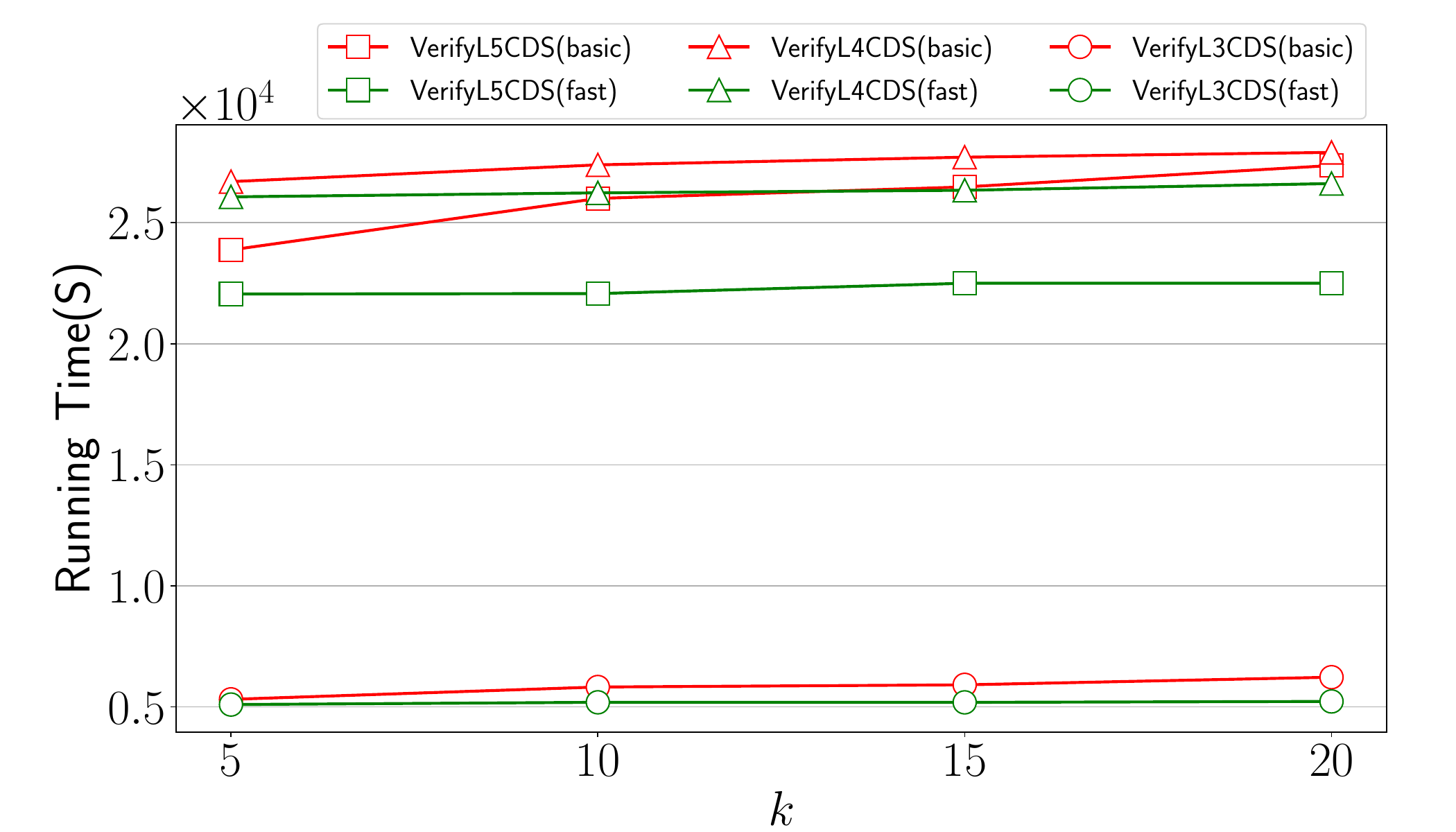}}
	\subfigure[{CA-CondMat}]{\includegraphics[width=0.24\textwidth]{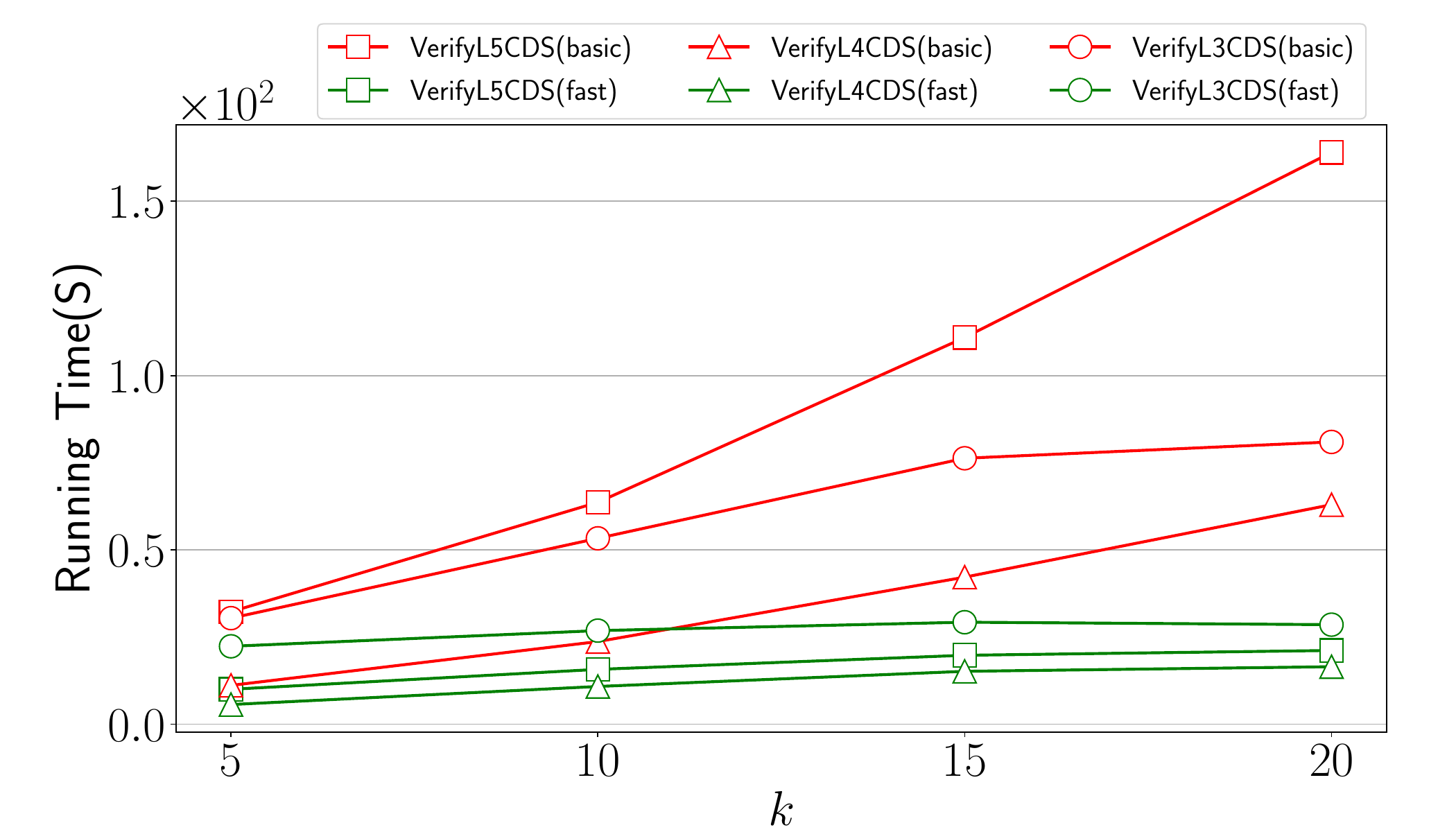}}
	\subfigure[{CA-GrQc}]{\includegraphics[width=0.24\textwidth]{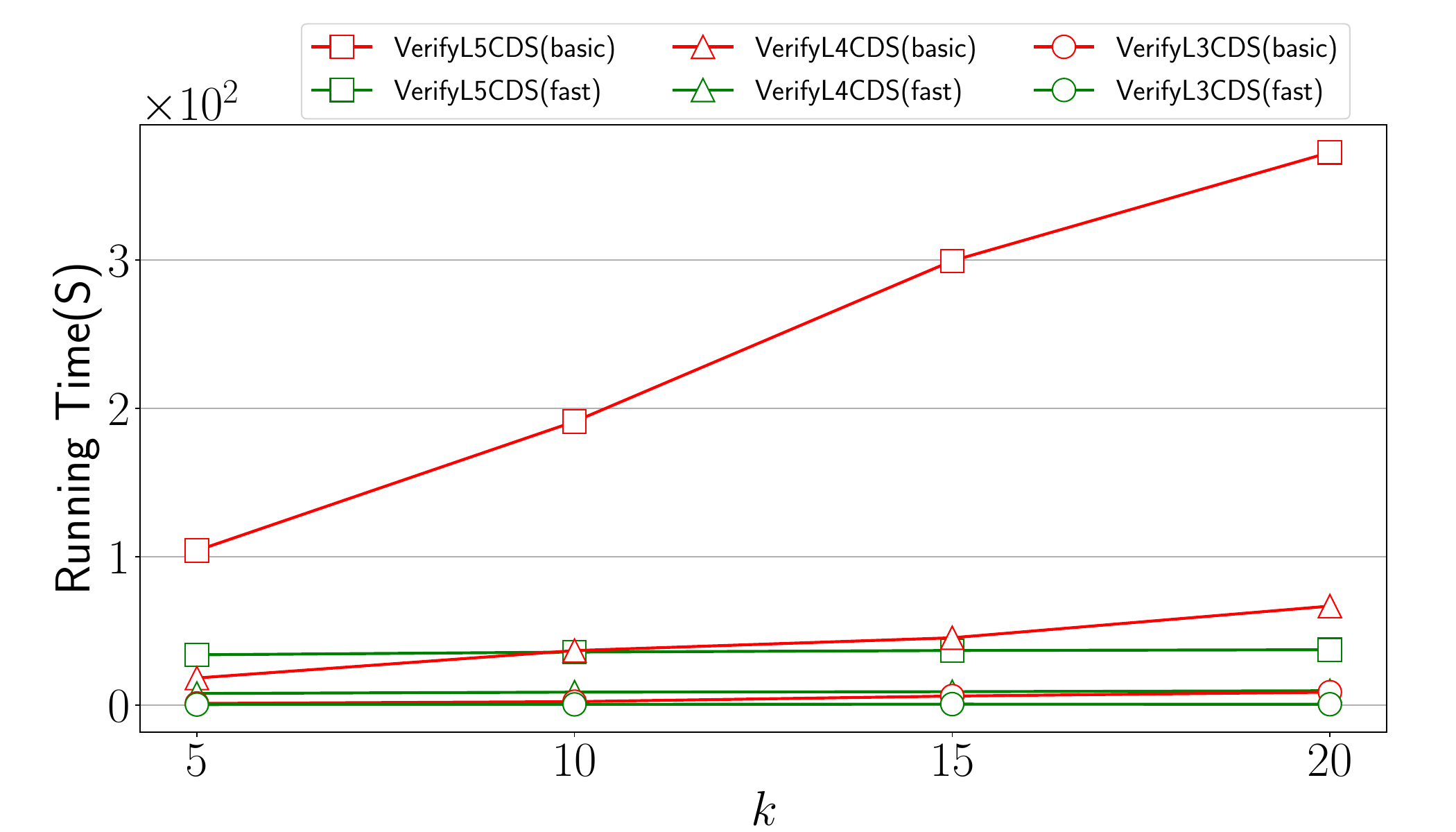}}
	\subfigure[{Amazon}]{\includegraphics[width=0.24\textwidth]{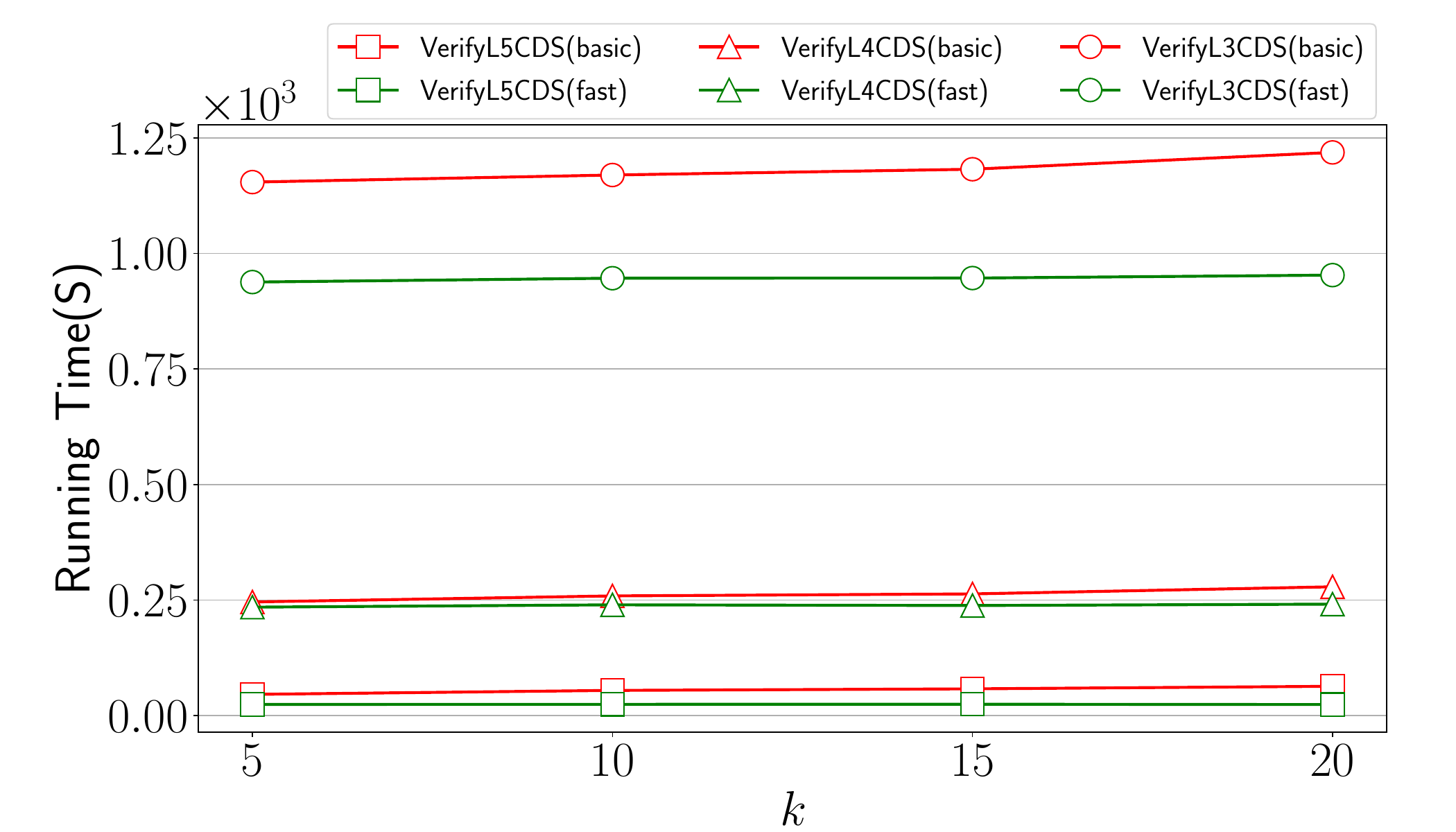}}
	\caption{Running time of algorithms with different $h$ (= 3,4,5) and $k$. \textcolor{red}{Red} is VerifyL$h$CDS(basic), \textcolor{myGreen}{Green} is VerifyL$h$CDS(fast)}
	\label{fig:running_time}
\end{figure*}
\begin{figure*}[htbp]
	\centering
	\includegraphics[width=0.96\textwidth]{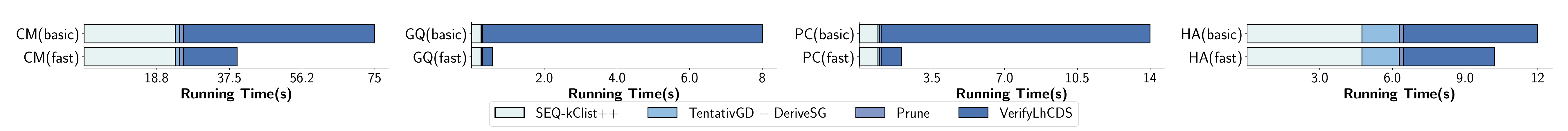}
	\caption{Running time of each part of \texttt{IPPV} with $h=3$ and $k=20$}
	\label{fig:time_proportion}
\end{figure*}

Here, we utilize our ``iterative propose-prune-and-verify'' pipeline to solve the L$hx$PDS problem. To apply the $hx$-pattern subgraph, there are some differences between Algorithm \ref{alg:LhCDS} and Algorithm \ref{alg:pattern} in the algorithmic details. In Algorithm \ref{alg:pattern}, we need to count $hx$-pattern graphs for Seq-kClist++ algorithm and derive candidate L$hx$PDS algorithm. In pruning part, the computation of $hx$-pattern graph cores is different for diverse kinds of patterns. Unlike $h$-clique, there may be more than one $hx$-pattern on a graph with $h$ vertices. In verification part, the methods for reducing the size of subgraph to compute the min-cut need small adjustments for different patterns. In general, the process of extending our algorithm to general patterns is concise and clear. In addition, our method may also support expressive graph models as directed graph, attributed graph, etc. These graph models also have meaningful patterns, such as directed closed loop in directed graph. The difficulty of generalizing \texttt{IPPV} to different graph models is related to the properties of the models.

\section{Experiments}
\subsection{Experimental Setup}
The datasets we use are undirected real-world graphs \cite{snapnets,nr}, including social networks, biological networks, web graphs, and collaboration networks. All datasets are listed in Table \ref{tab:setup}.

We compare the performances of the following algorithms: \\
\indent\textbf{\texttt{IPPV} :} the top-$k$ L$h$CDS discovery algorithm proposed by us.  \\
\indent\textbf{\texttt{LTDS} \cite{samusevich2016local} :} the top-$k$ LTDS discovery algorithm based on the maximum-flow, which solves the L$h$CDS problem with $h=3$.\\
\indent\textbf{\texttt{LDSflow} \cite{qin2015locally} :} the top-$k$ LDS discovery algorithm based on the maximum-flow, which solves the L$h$CDS problem with $h=2$.\\
\indent\textbf{\texttt{Greedy} :} the top-$k$ CDS discovery algorithm based on \texttt{KClist++} \cite{sun2020kclist++} using greedy approach. It has no guarantee on the locally densest property.

\footnotesize
\begin{table}[htbp]
	\caption{Datasets used in our experiments}
	\label{tab:setup}
	\begin{tabular}{l|l|l|l|l|l}
		\toprule
		\textbf{Name} & \textbf{Abbr.} & \textbf{$|V|$} & \textbf{$|E|$} & \textbf{$|\Psi_{3}|$} & \textbf{$|\Psi_{5}|$}\\ 
		\midrule
		soc-hamsterster & HA & 2,426 & 16,630 & 53,251 & 298,013\\
		CA-GrQc & GQ & 5,242 & 14,484 & 48,260 & 2,215,500\\
		fb-pages-politician & PP & 5,908 & 41,706 & 174,632 & 2,002,250\\
		fb-pages-company & PC & 14,113 & 52,126 & 56,005	& 207,829\\
		web-webbase-2001 & WB & 16,062 & 25,593 & 21,115 & 382,674\\
		CA-CondMat & CM & 23,133 & 93,439 & 173,361 & 511,088\\
		soc-epinions & EP & 26,588 & 100,120 & 159,700 & 521,106\\
		Email-Enron& EN& 36,692 & 183,831& 727,044 & 5,809,356\\
		loc-gowalla&GW &196,591& 950,327& 2,273,138 & 14,570,875\\
		
		DBLP & DB & 317,080 & 1,049,866 & 2,224,385 & 262,663,639\\
		Amazon & AM & 334,863 & 925,872 & 667,129 & 61,551\\
		
		soc-youtube & YT &495,957 & 1,936,748 & 2,443,886 & 5,306,643 \\
		soc-lastfm & LF &1,191,805 & 4,519,330 & 3,946,207& 10,404,656
		\\
		soc-flixster& FX &2,523,386 & 7,918,801 & 7,897,122&96,315,278 \\
		soc-wiki-talk &WT& 2,394,385 &	4,659,565 &9,203,519 & 382,777,822 \\
		\bottomrule
	\end{tabular}
\end{table}
\normalsize

All algorithms are implemented in C++ and compiled by g++ compiler at -O3 optimization level. All experiments are evaluated on a machine with Intel(R) Xeon(R) CPU 3.20GHz processor and 128GB memory, with Ubuntu operating system. Algorithms running for more than 48 hours are forcibly terminated.

\subsection{Efficiency under Parameter Variations}
In this section, we summarize the influence of different parameter changes on the running time. 
\subsubsection{Efficiency improvement by fast verification algorithm} We use VerifyL$h$CDS(basic) to represent the \texttt{IPPV} algorithm with Algorithm \ref{alg:compact} and VerifyL$h$CDS(fast) to represent the \texttt{IPPV} algorithm with Algorithm \ref{alg:verify}. Their running times are compared in Figure \ref{fig:running_time}. 
The fast verification algorithm with a smaller flow network is much faster than basic verification method. Especially as $k$ increases, the efficiency gap between the two algorithms becomes more apparent. We also compare the running time of the two verification algorithms in the total running time in Figure \ref{fig:time_proportion}. The acceleration effect of the fast algorithm is obvious. The results demonstrate the importance and benefit of optimizing the verification algorithm.

\subsubsection{Running time trends with varying $k$}
Parameter $k$ has a more pronounced impact on the running time of the algorithm than $h$. The experiments in Figure \ref{fig:running_time} indicate a direct relationship, where an increase in $k$ corresponds to a proportional increase in execution time. This trend is consistently observed across different datasets, which strengthens the premise that $k$ is an important factor in computational complexity. The running time of both algorithms increases significantly for incremental values of $k$. The only deviation is observed in the Email-Enron dataset, where the running time remains relatively static despite changes in $k$, due to the fact that the total number of L$h$CDSes in this dataset is smaller than $k$.

\subsubsection{Running time trends with varying $h$}
We took $h=3, 4, 5$ to compare the impact of $h$ on the running time. The results are shown in Figure \ref{fig:running_time}. When $h=5$, the running time is generally longer. The reason is that when $h=5$, the number of $5$-cliques is larger, as shown in Table \ref{tab:setup}. On Amazon, the running time is shorter when $h=5$, because the number of $5$-cliques is smaller. The running time is proportional to the number of $h$-cliques with different $h$.

\subsubsection{Influence of structural property}
We discuss how the density property ($|E|/|V|$) of graph datasets influences the efficiency of \texttt{IPPV}. As shown in Figure \ref{fig:varydensity}, we select four datasets and randomly sample edges of these datasets in different proportions. According to the outputs from these synthetically
generated graphs, the running time increases as the density of datasets increases. The reason is that when density is higher, the number of $h$-cliques increases, lengthening the total running time.

\begin{figure}[h]
	\centering
	\includegraphics[width=0.8\linewidth]{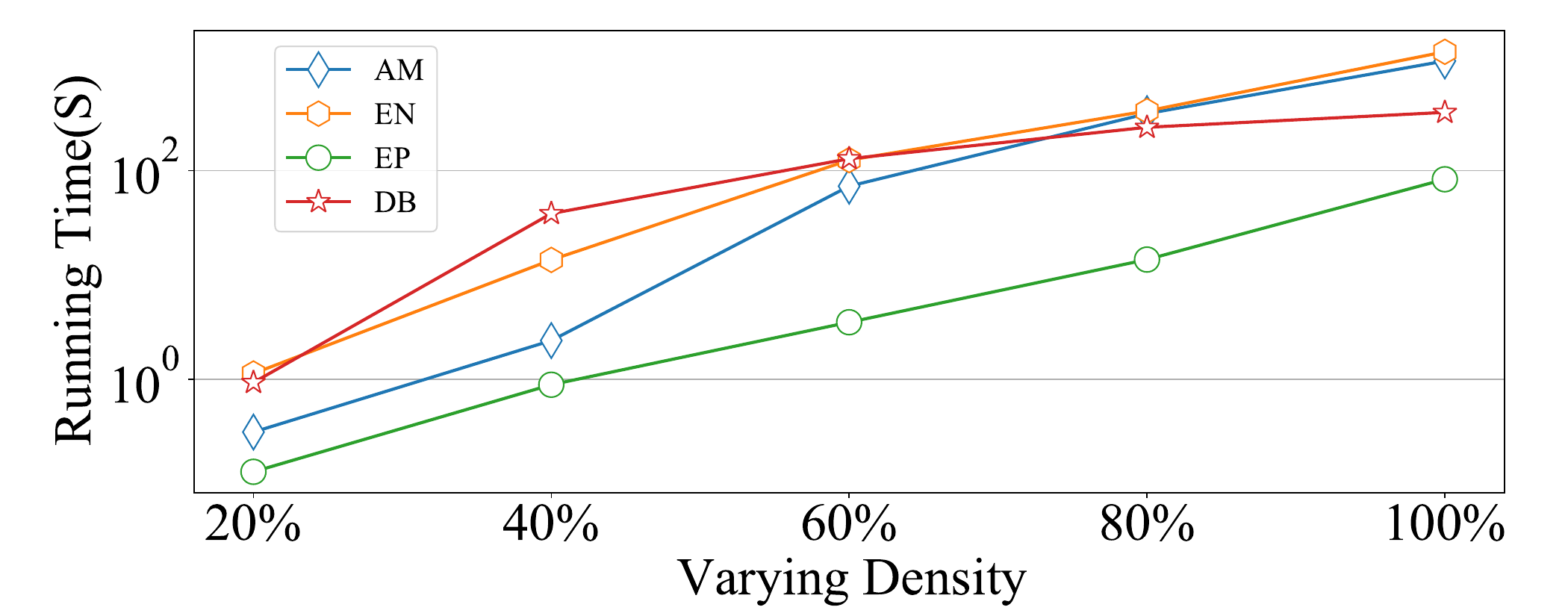}
	\caption{The running time on four datasets with varying density when $h=3$ and $k=5$}
	\label{fig:varydensity}
\end{figure}

\subsection{Efficiency v.s. Existing Algorithms}
Since LDSflow \cite{qin2015locally} outputs L$h$CDSes with $h = 2$, and LTDS \cite{samusevich2016local} outputs L$h$CDSes with $h = 3$, we compare IPPV with them by setting $h = 2$ and $h = 3$ respectively. 

\subsubsection{Efficiency: \texttt{IPPV} v.s. \texttt{LDSflow}}
We set $h=2$, $k=5$ to observe the running time of the two algorithms. The results are shown in Figure \ref{fig:ldsflow}. 
We compare \texttt{IPPV} with \texttt{LDSflow} on eight datasets, and \texttt{IPPV} has efficiency improvements on all the datasets. The bottleneck of \texttt{LDSflow} is the loose upper and lower bounds. 
\begin{figure}[h]
	\centering
	\includegraphics[width=0.8\linewidth]{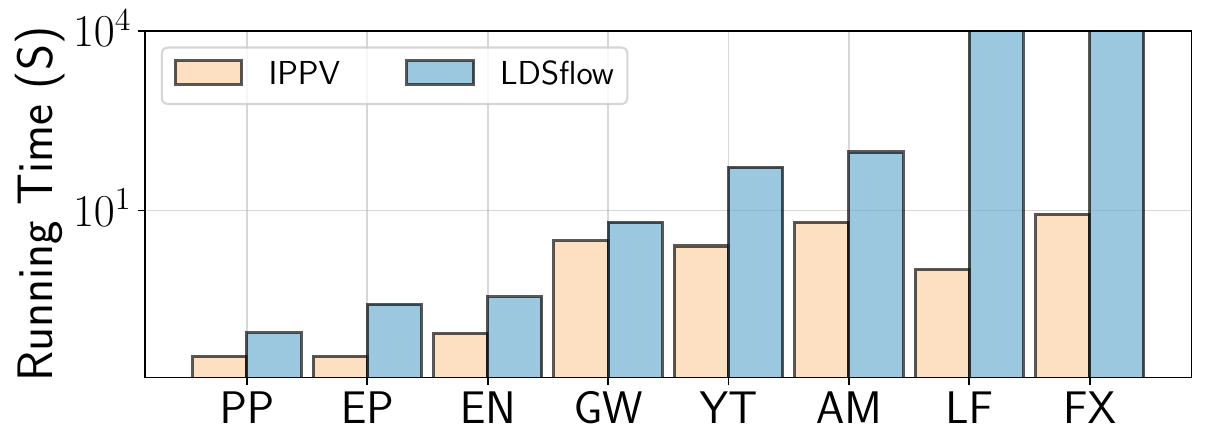}
	\caption{Efficiency of \texttt{IPPV} (h=2) and \texttt{LDSflow}}
	\label{fig:ldsflow}
\end{figure}

\subsubsection{Efficiency: \texttt{IPPV} v.s. \texttt{LTDS}}
We set $k=5$ and $h=3$ in experiments and the results are shown in Table \ref{tab:runningtime}. \texttt{IPPV} and \texttt{LTDS} are compared on all the datasets. There are significant efficiency improvements on all datasets. The main bottleneck of LTDS is the time-consuming verification part, and the reason is that the upper and lower bounds of \texttt{LTDS} are not as tight as that of \texttt{IPPV}, so there will be more failures in the verification part. The running time of our algorithm is closely related to the size of the graph and the number of $h$-cliques. A rise in the number of $h$-cliques can cause an increase in running time.
\footnotesize
\begin{table}[htbp]
	\caption{Efficiency of \texttt{IPPV} (h=3) and \texttt{LTDS}}
	\label{tab:runningtime}
	\begin{tabular}{l|l|l|l}
		\toprule
		\textbf{Dataset} & \textbf{IPPV (h=3)} & \textbf{LTDS} & \textbf{Speedup}\\ 
		\midrule
		
		soc-hamsterster  & \cellcolor{red!20}\textbf{7.50(s)} & 46.54 & 6.20$\times$ \\
		CA-GrQc & \cellcolor{red!20}\textbf{0.38} & 18.97 &  49.92 $\times$ \\
		fb-pages-politician & \cellcolor{red!20}\textbf{32.32} & 436.30 & 13.50$\times$   \\
		fb-pages-company & \cellcolor{red!20}\textbf{2.56} & 51.48 & 20.11$\times$  \\
		web-webbase-2001 & \cellcolor{red!20}\textbf{0.14} & 12.20 & 87.14$\times$  \\
		CA-CondMat & \cellcolor{red!20}\textbf{21.63} & 541.63 & 25.04$\times$  \\
		soc-epinions & \cellcolor{red!20}\textbf{82.54} & 558.91 & 6.77$\times$  \\
		Email-Enron& \cellcolor{red!20}\textbf{1369.84} & 2253.14 & 1.64$\times$  \\
		loc-gowalla& \cellcolor{red!20}\textbf{5095.63} & 68216.14 & 13.39$\times$  \\
		DBLP & \cellcolor{red!20}\textbf{360.49} & 4888.93 & 13.56$\times$  \\
		Amazon & \cellcolor{red!20}\textbf{1118.08} & 1308.53 &  1.17$\times$  \\
		soc-youtube & \cellcolor{red!20}\textbf{9070.89} & 42821.99 & 4.72$\times$  \\
		soc-lastfm & \cellcolor{red!20}\textbf{11223.13} & $\ge$172, 800 & $\ge$ 15.40 $\times$ \\
		soc-flixster& \cellcolor{red!20}\textbf{3018.62} & $\ge$ 172, 800 &  $\ge$ 57.24 $\times$ \\
		soc-wiki-talk & \cellcolor{red!20}\textbf{57382.42} & $\ge$ 172, 800 & $\ge$ 3.011 $\times$ \\
		\bottomrule
	\end{tabular}
\end{table}
\normalsize

\subsection{Characteristics of the Detected L$h$CDSes}
We adopt different quality measures to show the characteristics of the detected L$h$CDSes, and then we compare our detected L$h$CDSes with those detected by the \texttt{Greedy} algorithm.

\subsubsection{Visualization of  L$h$CDSes with varying $h$}
We use a network of books about US politics which were sold by Amazon \cite{Krebs2004Books} to visualize the characteristics of detected L$h$CDSes with varying $h$. 
The vertices represent different books. Figure \ref{fig:clique_case}(a) shows the books fall into neutral(green), liberal(blue), and conservative(red) categories. The edges represent frequent co-purchasing of books by the same buyers, which indicate ``customers who bought this book also bought the other books''  on Amazon. 
Figure \ref{fig:clique_case}(b)-Figure \ref{fig:clique_case}(e) visualize the detected L$h$CDSes with varying $h$.
The set of steelblue vertices is the top-$1$ L$h$CDS, and the set of orange vertices if exists, is the top-$2$ L$h$CDS. The vertices in the dataset have clear attributes, which is convenient for us to measure whether \texttt{IPPV} has the ability to mine diverse communities. As visualized by the results, L$h$CDSes with larger $h$ are closer to a clique. Besides, when $h$ is larger, L$h$CDSes can find multiple dense communities in different categories. L$4$CDSes contain both liberal and conservative book communities, whereas LDSes (i.e., L$2$CDSes) only contain liberal book community. The results show the potential of \texttt{IPPV} to mine diverse categorical communities.

\begin{figure}[htbp]
	\centering
	\subfigure[categories]{\includegraphics[width=0.18\linewidth]{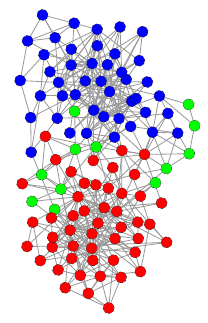}\label{fig:examplelabel}}
	\subfigure[$h=2$]{\includegraphics[width=0.18\linewidth]{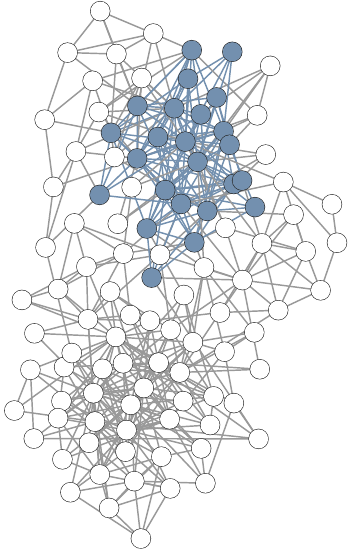}\label{fig:pol2_clique}}
	\subfigure[$h=3$]{\includegraphics[width=0.18\linewidth]{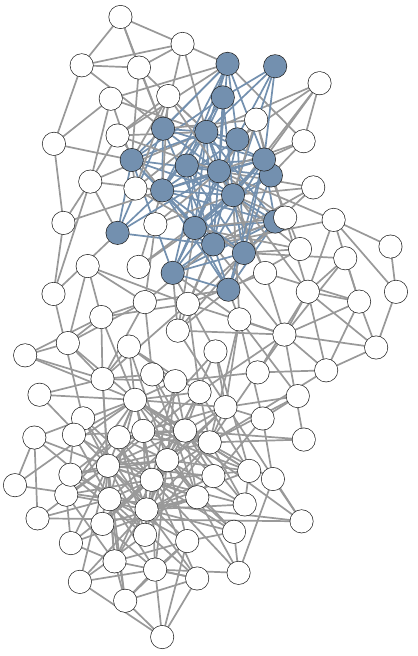}\label{fig:pol3_clique}}
	\subfigure[$h=4$]{\includegraphics[width=0.18\linewidth]{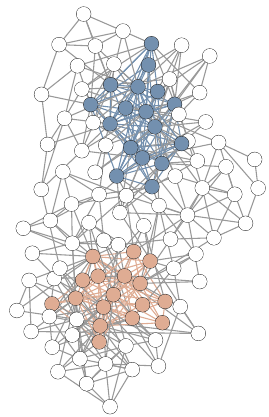}\label{fig:pol4_clique}}
	\subfigure[$h=5$]{\includegraphics[width=0.18\linewidth]{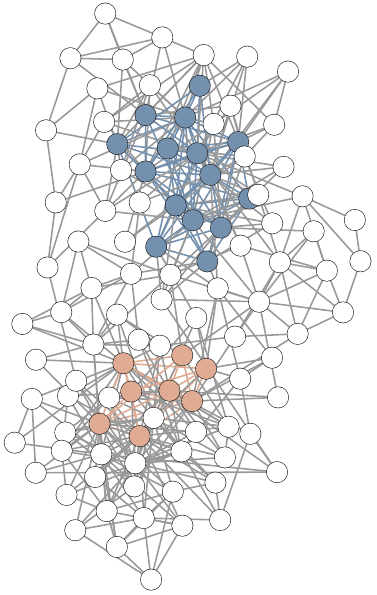}\label{fig:pol5_clique}}
	\caption{L$h$CDS case study on real network (the top-$1$ L$h$CDS: steelblue; the top-$2$ L$h$CDS: orange vertices)}
	\label{fig:clique_case}
\end{figure}

\subsubsection{Edge density of L$h$CDSes with varying $h$}
We compare the average edge density ($(2\times|E|)/[|V|\times(|V|-1)]$) of top-$5$ L$h$CDSes for varying $h$ in Table \ref{tab:avgd}. When $h$ is larger, L$h$CDSes generally have higher edge density. This is aligned with our general purpose.  
\footnotesize
\begin{table}[H]
	\caption{Average edge density and diameter with varying $h$}
	\label{tab:avgd}
	\resizebox{\linewidth}{!}{
		\begin{tabular}{c|c|cccc|c|cccc}
			\hline
			& \multicolumn{5}{c
				|}{Average Edge Density} &  \multicolumn{5}{c}{Average Diameter}                                                         \\ \cline{2-11} 
			\multirow{-2}{*}{dataset} & $h$=2 & $h$=3   & $h$=5                          & $h$=7                          & $h$=9    &  $h$=2 & $h$=3   & $h$=5                          & $h$=7                          & $h$=9                      \\ \hline
			PC          & \cellcolor{red!10}0.752    & \cellcolor{red!20}0.858   & \cellcolor{red!30}0.891                          & \cellcolor{red!40}0.894                         & \cellcolor{red!50}\textbf{0.957 } &2.00 &2.00 	&2.00 	&2.00 	&2.00   \\
			HA           & \cellcolor{red!10}0.805    & \cellcolor{red!20}0.869   & \cellcolor{red!40}0.987                          & \cellcolor{red!30}0.982                         & \cellcolor{red!50}\textbf{0.995 } &1.80	&1.60	&1.20	&1.40	&1.40 \\
			PP       & \cellcolor{red!30}0.709     & \cellcolor{red!10}0.685   & \cellcolor{red!20}0.696                         & \cellcolor{red!40}0.729                         & \cellcolor{red!50}\textbf{0.827 } &2.20	&2.00	&2.00	&2.00	&2.00 \\
			CM                & \cellcolor{red!10}0.984    & \cellcolor{red!20}0.984   & \cellcolor{red!30}0.986  & \cellcolor{red!40}0.986                         & \cellcolor{red!50}\textbf{0.986 }   &1.20	&1.20	&1.40	&1.40	&1.40       \\
			EP              & \cellcolor{red!10}0.487    & \cellcolor{red!20}0.724   & \cellcolor{red!50}\textbf{0.799}                        & \cellcolor{red!40}0.770                         & \cellcolor{red!30}0.755  &2.60	&1.80	&1.80	&1.67	&2.00 \\
			WB          & \cellcolor{red!10}0.987     & \cellcolor{red!20}0.987  & \cellcolor{red!30}0.997                         & \cellcolor{red!40}0.997  & \cellcolor{red!50}\textbf{0.997 }     &1.40	&1.40	&1.20	&1.20	&1.20      \\
			GQ                   & \cellcolor{red!10}0.972    & \cellcolor{red!20}0.972  & \cellcolor{red!30}0.972  & \cellcolor{red!40}\textbf{0.975}                        & OOM      &1.60	&1.60	&1.60	&1.60	&OOM          \\ \hline
		\end{tabular}
	}
\end{table}
\normalsize
\subsubsection{Diameter of L$h$CDSes with varying $h$}
We list the average diameter (the longest distance of all pairs of nodes in a graph) of top-$5$ L$h$CDSes of different $h$ in Table \ref{tab:avgd}. It is obvious that the diameters of the subgraphs produced by L$h$CDS ($h\geq 3$) do not exceed $2$. Coupled with the edge density results, while L$h$CDS can find subgraphs with relatively higher edge density compared to LDS, the diameters of the subgraphs found by L$h$CDS are as small as those found by LDS. This suggests that L$h$CDSes are cohesive, and the nodes within L$h$CDSes are highly interconnected.

\begin{figure}[h]
	\centering
	\includegraphics[width=\linewidth]{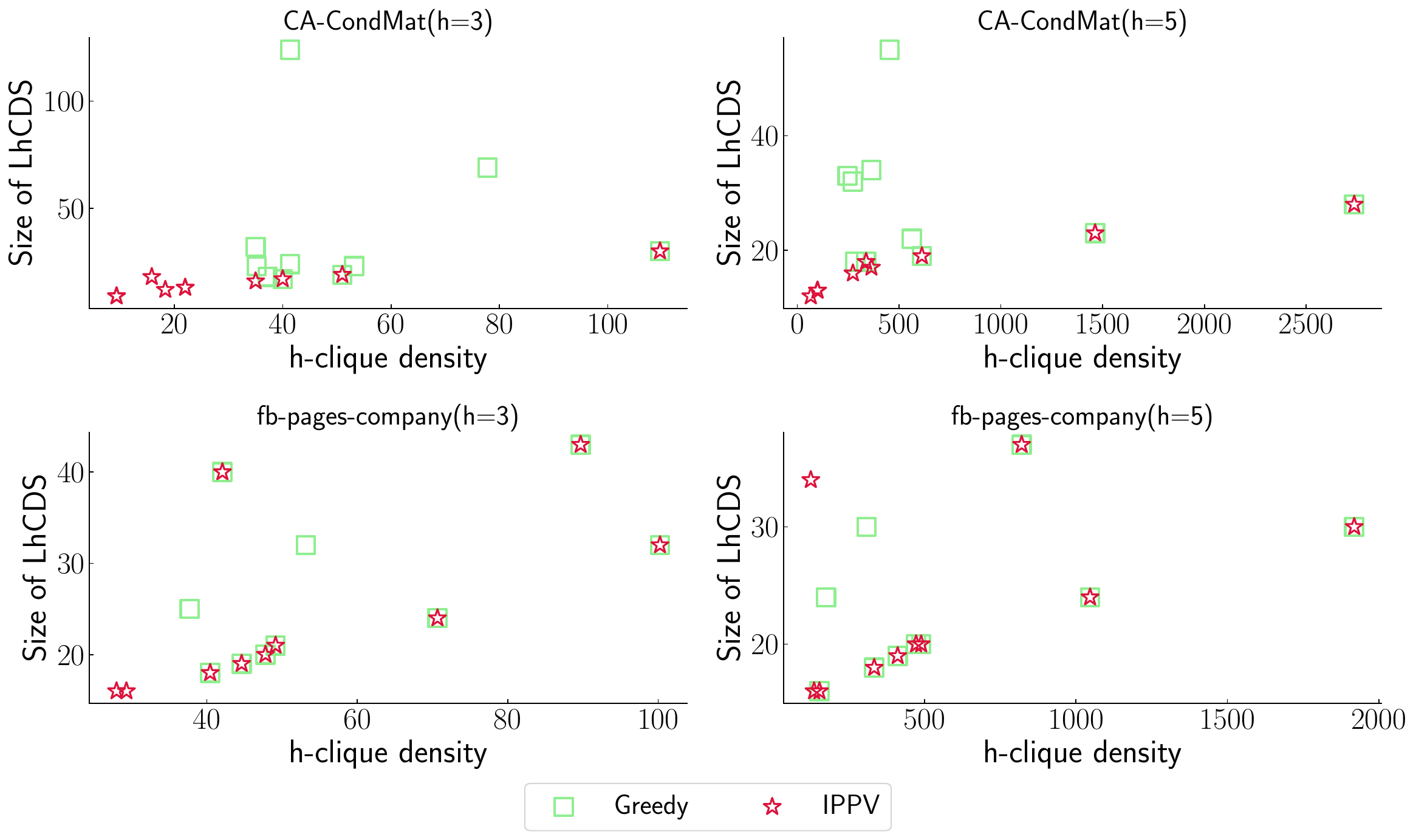}
	\caption{Subgraph statistics of $h$-clique density and size}
	\label{fig:exp2}
\end{figure}

\subsubsection{Size and $h$-clique density of subgraphs: \texttt{IPPV} vs \texttt{Greedy}}
Next, we compare the L$h$CDSes detected by \texttt{IPPV} and the $h$-clique densest subgraphs found by the \texttt{Greedy} algorithm. We select $h=3, 5$ on two datasets, and the results are shown in Figure \ref{fig:exp2}.
First, the results of the two algorithms overlap to a certain extent, among which the top-$1$ CDS is the same because the first L$h$CDS must be the $h$-clique densest subgraph in the whole graph. 
Second, there is a certain difference between the returned subgraphs of the \texttt{Greedy} algorithm and \texttt{IPPV}. For example, in CA-CondMat ($h$=3), the $h$-clique density of the second output subgraph is 51 for \texttt{IPPV}, and 78 for \texttt{Greedy}, but the returned subgraphs of \texttt{Greedy} is adjacent to the first output subgraph without the guarantee of the locally densest property. Therefore, the \texttt{Greedy} algorithm cannot solve the L$h$CDS problem well. The two algorithms totally overlap if and only if the top-$k$ $h$-clique densest subgraph belongs to different regions occasionally.

\subsection{Clustering Coefficient of Different $h$}
Since near clique is an important criterion for evaluating dense subgraphs, we evaluate how  L$h$CDSes of different $h$ are close to the clique structure. In graph theory, \emph{clustering coefficient} is a measure of the degree to which vertices in a graph tend to cluster together,  which is a direct measure to the degree of near clique. For each vertex $u\in V$, which has $k_{u}$ neighbors $N_{u}$ ($|N_{u}| = k_{u}$), the clustering coefficient of $u$ is  $C_{u}=\frac{2|\{e_{vw}: v,w\in N_{u}, e_{vw}\in E\}|}{k_{u}(k_{u}-1)}$. We compare the average $C_{u}$ of all the L$h$CDSes of different $h$ in Table \ref{tab:avgcc}.

\footnotesize
\begin{table}[H]
	\caption{Average clustering coefficient of different $h$ values}
	\label{tab:avgcc}
	\resizebox{\linewidth}{!}{
		\begin{tabular}{c|c|cccc}
			\hline
			& \multicolumn{5}{c}{Average Clustering Coefficient}                                                            \\ \cline{2-6} 
			\multirow{-2}{*}{dataset} & $h$=2 & $h$=3   & $h$=5                          & $h$=7                          & $h$=9                          \\ \hline
			fb-pages-company          & \cellcolor{red!10}0.582    & \cellcolor{red!20}0.852  & \cellcolor{red!30}0.895                         & \cellcolor{red!40}0.915                         & \cellcolor{red!50}\textbf{0.930} \\
			soc-hamsterster           & \cellcolor{red!10}0.480    & \cellcolor{red!20}0.910  & \cellcolor{red!40}0.990                         & \cellcolor{red!30}0.984                        & \cellcolor{red!50}\textbf{0.995} \\
			fb-pages-politician       & \cellcolor{red!10}0.583    & \cellcolor{red!20}0.683  & \cellcolor{red!30}0.776                        & \cellcolor{red!40}0.798                        & \cellcolor{red!50}\textbf{0.835} \\
			CA-CondMat                & \cellcolor{red!10}0.567    & \cellcolor{red!20}0.977  & \cellcolor{red!50}\textbf{0.992} & \cellcolor{red!50}0.992                         & \cellcolor{red!40}0.991                         \\
			soc-epinions              & \cellcolor{red!10}0.231    & \cellcolor{red!40}0.722  & \cellcolor{red!20}0.701                        & \cellcolor{red!30}0.705                         & \cellcolor{red!50}\textbf{0.773} \\
			web-webbase-2001          & \cellcolor{red!10}0.831    & \cellcolor{red!20}0.884 & \cellcolor{red!40}0.989                        & \cellcolor{red!50}\textbf{0.992} & \cellcolor{red!30}0.979                        \\
			CA-GrQc                   & \cellcolor{red!10}0.533    & \cellcolor{red!20}0.975 & \cellcolor{red!30}0.982 & \cellcolor{red!40}\textbf{0.985}                       & OOM               \\ \hline
		\end{tabular}
	}
\end{table}
\normalsize

According to the results shown in Table \ref{tab:avgcc}, when $h$ is larger, the average $C_{u}$ is generally larger, showing the L$h$CDSes with larger $h$ are closer to clique. In addition, there is a big difference between $h=3$ and $h=2$ (L$2$CDS is LDS), which shows that LDS is less dense than other L$h$CDS. Our algorithm is important for finding near-clique subgraphs, which cannot be replaced by LDS.

\subsection{Memory Overheads}
We compare the memory utilization for the \texttt{IPPV} and \texttt{LTDS} algorithms across all datasets ($h = 3$, $k = 5$).
Figure \ref{fig:exp3} illustrates a clear correlation between memory usage and dataset size. 
\texttt{IPPV} strategically reduces the size of candidate subgraphs through a pruning mechanism prior to evaluating self-compactness. The verifying part often dominates the memory consumption.

\begin{figure}[h]
	\centering
	\includegraphics[width=0.8\linewidth]{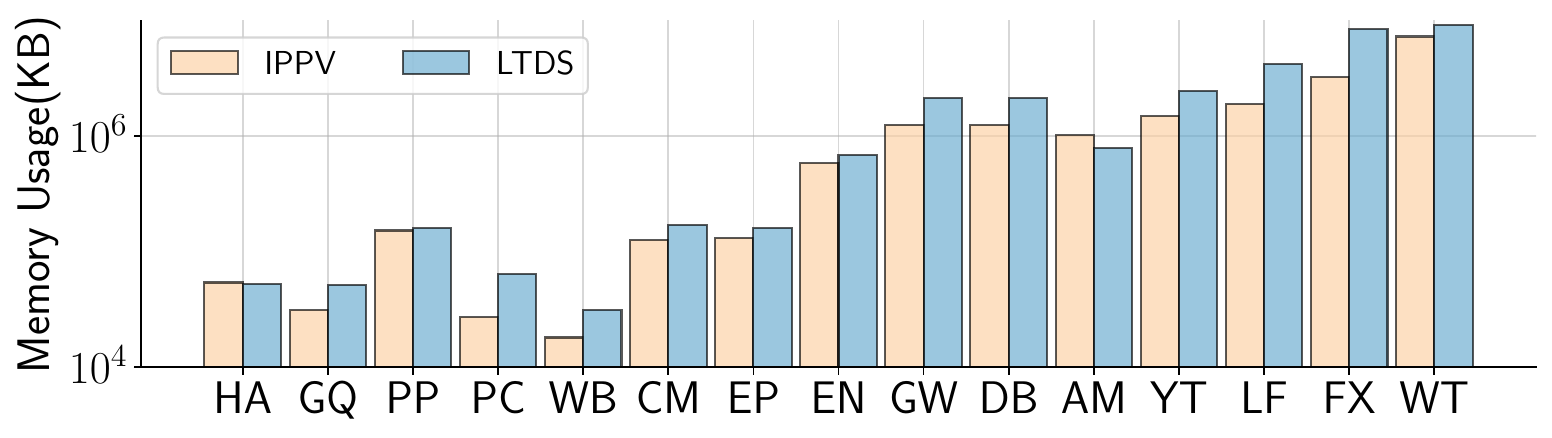}
	\caption{Memory usage of algorithms}
	\label{fig:exp3}
\end{figure}

\subsection{The Number of Iterations}
To choose the optimal number of iterations $T$ of \texttt{SEQ-kClist++}, we set different $T$ on the \texttt{IPPV} algorithm. We select $T=5, 10, 15, 20, 40,$ $60, 80, 100$, as shown in Figure \ref{fig:varyt}. The experiment on eight datasets shows that the optimal performance is between $15$ and $20$ iterations. In our experiments, we choose $T=20$.
\begin{figure}[h]
	\centering
	\includegraphics[width=\linewidth]{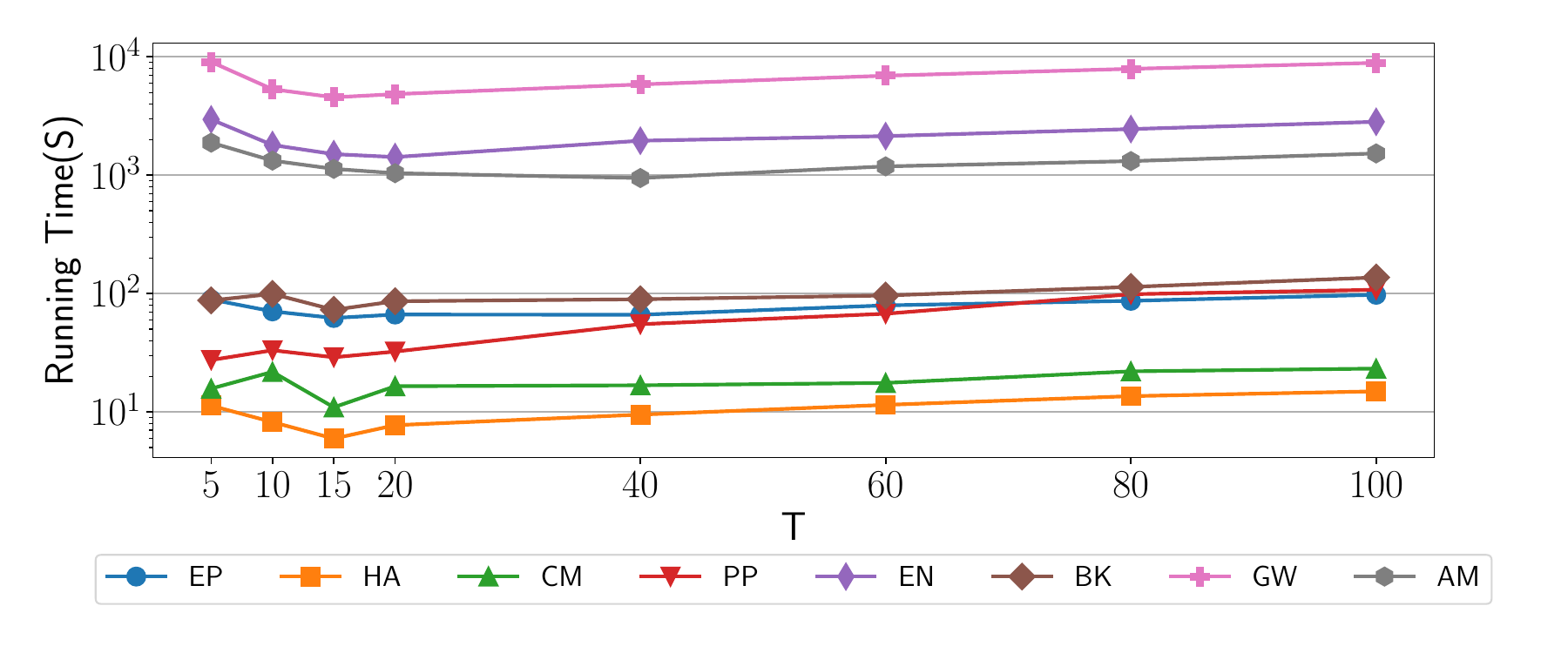}
	\caption{Running time of eight datasets with different $T$}
	\label{fig:varyt}
\end{figure}

\subsection{Case Study of L$hx$PDS}
We utilize the same real dataset \cite{Krebs2004Books} to experimentally illustrate the L$hx$PDS problem. For each pattern depicted in Figure \ref{fig:pattern}, we compute the results of L$4x$PDS. In Figure \ref{fig:pattern_case}, the set of steelblue vertices is the top-$1$ L$hx$PDS, and the set of orange vertices if exists, is the top-$2$ L$hx$PDS of the pattern $hx$. It is evident that the L$4x$PDS corresponding to various patterns exhibit differences in terms of the number of L$4x$PDS, the number of vertices, and the position of vertices. 
To delve deeper into graph analysis, tasks such as community clustering can be extended to explore the L$hx$PDS subgraph.
\begin{figure}[htbp]
	\centering
	\subfigure[$3$-star]{
		\includegraphics[width=0.14\linewidth]{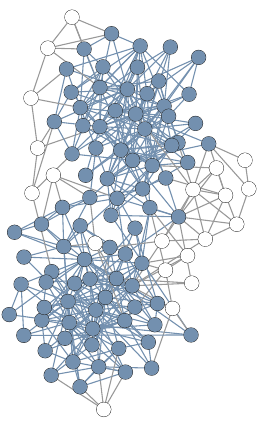}
		\label{fig:pol_3_star}
	}
	\subfigure[$4$-path]{
		\includegraphics[width=0.14\linewidth]{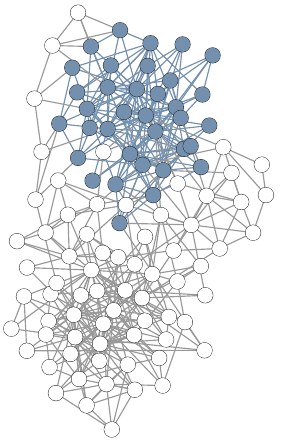}
		\label{fig:pol_4_path}
	}
	\subfigure[c$3$-star]{
		\includegraphics[width=0.14\linewidth]{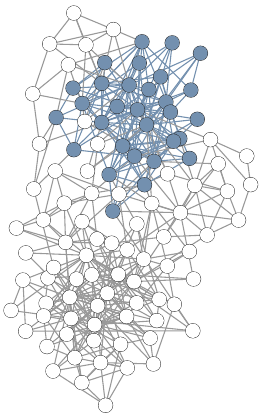}
		\label{fig:pol_c3_star}
	}
	\subfigure[$4$-loop]{
		\includegraphics[width=0.14\linewidth]{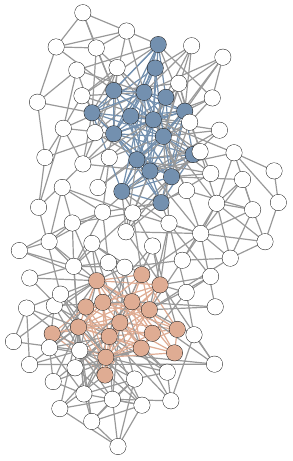}
		\label{fig:pol_4loop}
	}
	\subfigure[$2$-triangle]{
		\includegraphics[width=0.14\linewidth]{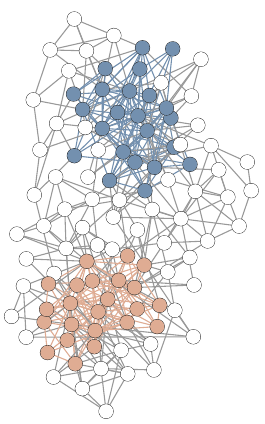}
		\label{fig:pol2_triangle}
	}
	\subfigure[$4$-clique]{
		\includegraphics[width=0.14\linewidth]{figure/4cliques.pdf}
		\label{fig:pol_4_clique}
	}
	\caption{L$4x$PDS case study on real network (the top-$1$ L$hx$PDS: steelblue; the top-$2$ L$hx$PDS: orange vertices)}
	\label{fig:pattern_case}
\end{figure}

\section{CONCLUSION}
In this paper, we study how to discover locally $h$-clique densest subgraphs in a graph $G$, i.e., the L$h$CDS problem. We present \texttt{IPPV}, an iterative propose-prune-and-verify pipeline for top-$k$ L$h$CDS detection. The $h$-clique compact number bounds and graph decomposition method which help to derive L$h$CDS candidates efficiently are proposed. A new optimized verification algorithm is designed, and its correctness is proved. The extension of our algorithm to solve the locally general pattern densest subgraph problem is feasible and promising. Extensive experiments on real datasets show the high efficiency and scalability of our proposed algorithm. When $h$ is large in large-scale graphs, there is still room for further optimizing \texttt{IPPV}. 
We will continue to optimize the algorithm and further explore the L$hx$PDS problem in our future work.

\begin{acks}
Dr. Wang is supported in part by the National Natural Science Foundation of China Grant No. 61972404, Public Computing Cloud, Renmin University of China, and the Blockchain Lab. School of Information, Renmin University of China. Dr. Li is supported in part by the National Natural Science Foundation of China Grant No.12071478. The authors are grateful to Professor Lijun Chang for his help with the revision of this paper. The authors would like to thank the anonymous reviewers and shepherd for providing constructive feedback and valuable suggestions.
\end{acks}

\bibliographystyle{ACM-Reference-Format}
\bibliography{LhCDS}

\end{document}